\documentclass[conference]{IEEEtran}
\IEEEoverridecommandlockouts
\usepackage[noadjust]{cite}

\usepackage{amsmath,amssymb,amsfonts}
\usepackage{etoolbox}
\usepackage{algorithmic}
\usepackage{graphicx}
\usepackage{tikz}
\usepackage{pgfplots}
\pgfplotsset{compat=1.13}
\usepackage{float}
\usepackage{textcomp}
\usepackage{xcolor}
\usepackage[all]{xy} 
\usepackage{bm}
\usepackage{multirow}
\usetikzlibrary{shapes,snakes,patterns}

\pgfdeclarepatternformonly{soft horizontal lines}{\pgfpointorigin}{\pgfqpoint{100pt}{1pt}}{\pgfqpoint{100pt}{3pt}}%
{
	\pgfsetstrokeopacity{0.3}
	\pgfsetlinewidth{0.1pt}
	\pgfpathmoveto{\pgfqpoint{0pt}{0.5pt}}
	\pgfpathlineto{\pgfqpoint{100pt}{0.5pt}}
	\pgfusepath{stroke}
}

\pgfdeclarepatternformonly{soft crosshatch}{\pgfqpoint{-1pt}{-1pt}}{\pgfqpoint{4pt}{4pt}}{\pgfqpoint{3pt}{3pt}}%
{
	\pgfsetstrokeopacity{0.3}
	\pgfsetlinewidth{0.4pt}
	\pgfpathmoveto{\pgfqpoint{3.1pt}{0pt}}
	\pgfpathlineto{\pgfqpoint{0pt}{3.1pt}}
	\pgfpathmoveto{\pgfqpoint{0pt}{0pt}}
	\pgfpathlineto{\pgfqpoint{3.1pt}{3.1pt}}
	\pgfusepath{stroke}
}

\def\BibTeX{{\rm B\kern-.05em{\sc i\kern-.025em b}\kern-.08em
    T\kern-.1667em\lower.7ex\hbox{E}\kern-.125emX}}

\newcommand{\beq}{\begin{equation}}
\newcommand{\enq}{\end{equation}}
\newcommand{\bel}{\begin{lemma}}
\newcommand{\enl}{\end{lemma}}
\newcommand{\bet}{\begin{theorem}}
\newcommand{\ent}{\end{theorem}}

\newcommand{\suppress}[1]{}

\mathchardef\mhyphen="2D

\def\sM{\mathsf{M}}
\def\sL{\mathsf{L}}

\makeatletter
\newcommand*{\rom}[1]{\expandafter\@slowromancap\romannumeral #1@}
\makeatother

\mathchardef\mhyphen="2D

\newtheorem{remark}{Remark}
\newtheorem{theorem}{Theorem}
\newtheorem{lemma}{Lemma}
\newtheorem{corollary}{Corollary}
\newtheorem{proposition}{Proposition}

\makeatletter

\newcommand{\Rmnum}[1]{\expandafter\@slowromancap\romannumeral #1@}
\makeatother

\newcommand{\Co}{C}

\graphicspath{{./fig/}}
\DeclareGraphicsExtensions{.pdf} 

\def\QED{\mbox{\rule[0pt]{1.5ex}{1.5ex}}}

\def\Label#1{\label{#1}\ [\ \text{#1}\ ]\ }
\def\Label{\label}

\AtBeginEnvironment{equation}{\small}
\AtBeginEnvironment{equation*}{\small}
\AtBeginEnvironment{align}{\small}
\AtBeginEnvironment{align*}{\small}
\AtBeginEnvironment{gather}{\small}
\AtBeginEnvironment{gather*}{\small}
\AtBeginEnvironment{multline}{\small}
\AtBeginEnvironment{multline*}{\small}
\begin{document}

\title{Two-Way Wiretap Channel under Mixed Secrecy Constraint*\\
	\thanks{MH was supported in part by the General R\&D Projects of 1+1+1 CUHK-CUHK(SZ)-GDST Joint Collaboration Fund (No. GRDP2025-022) and the Guangdong Provincial Quantum Science Strategic Initiative (No. GDZX2505003).}
}

\author{\IEEEauthorblockN{Yanling Chen}
	\IEEEauthorblockA{		
		\textit{Volkswagen Infotainment GmbH, Germany}\\
		Bochum, Germany \\
		Email: yanling.chen@volkswagen-infotainment.com}
	\and
	\IEEEauthorblockN{Masahito Hayashi}
	\IEEEauthorblockA{
		\textsuperscript{1}\textit{School of Data Science, The Chinese University of Hong Kong}\\
		Shenzhen, China\\
		\textsuperscript{2}\textit{International Quantum Academy, Shenzhen, China}\\
		\textsuperscript{3}\textit{Graduate School of Mathematics, Nagoya University}\\
		Nagoya, Japan\\
		Email: hmasahito@cuhk.edu.cn}
}

\maketitle

\begin{abstract}
This paper studies the two-way wiretap channel (TW-WC) with an external eavesdropper under strong one-sided (mixed) secrecy: only User~1's message is required to be secure from the eavesdropper, while no secrecy constraint is imposed on User~2's message. Secrecy is measured by the information leakage of User~1's message to the eavesdropper. 
Applying a non-adaptive construction and a one-sided reduced key-exchange construction, we obtain exponential error and leakage bounds for every fixed number of adaptive rounds. A quantitative one-time-pad argument propagates the leakage of the key used in the next round, and an explicit padding construction removes the initialization-rate loss for all sufficiently large blocklengths. We derive the corresponding strong mixed-secrecy achievable regions. For every product input distribution satisfying the three strict feasibility conditions of the non-adaptive construction, its strong-secrecy inner bound contains the previously reported weak one-sided single-letter inner bound. The overall adaptive achievable region includes the non-adaptive construction as a fallback and also contains a key-exchange subregion that can be strictly larger. 
\end{abstract}

\begin{IEEEkeywords}
two-way wiretap channel, strong secrecy, adaptive coding, channel resolvability, R\'enyi mutual information.
\end{IEEEkeywords}

\section{Introduction}

The two-way communication channel (TWC) was first introduced by Shannon in \cite{Shannon1961}, where he focused on the discrete memoryless TWC and established an inner bound and an outer bound for the capacity region. In general, it is known that Shannon's inner bound does not coincide with Shannon's outer bound \cite{Dueck1979, Schalkwijk1983}. 
So far, the capacity region of the TWC has been determined only for some special cases \cite{Shannon1961, Han1984, Varshney2013, SAL2016, CVA2017}, while a single-letter characterization of the capacity region of a general TWC remains open. 

Shannon's inner bound is achieved by non-adaptive encoders whose inputs depend only on the messages and not on past outputs. In contrast, the outer bound allows dependent inputs, as may arise from adaptation to past outputs. The difficulty in characterizing the general capacity region is to identify the appropriate form of adaptation or input dependence \cite{ZBS1986, PW1989, TU2007}. 

Information theoretic secrecy was also introduced by Shannon but in a different pioneering paper \cite{Shannon1949}. In particular, he formulated the concept of perfect secrecy, where the eavesdropper's observation $\mathbf{Z}$ provides no information on the transmitted message $M$ (i.e., the information leakage to the eavesdropper $I(M;\mathbf{Z})=0$). On the other hand, Wyner in \cite{Wyner1975} proposed the wiretap channel, which is the one-way discrete memoryless channel with an external eavesdropper; and generalized the ``perfect secrecy" to an ``asymptotic perfect secrecy" by measuring the normalized equivocation at the eavesdropper about the message $M$ given the eavesdropper's observation $\mathbf{Z}.$ This normalized equivocation is equivalent to the information leakage rate to the eavesdropper, which is often referred to as "weak secrecy" in the literature, as it tolerates the fact that the eavesdropper might obtain a substantial amount of information in an absolute sense \cite{src:Csiszar1996,MH2006}, i.e.,  $I(M;\mathbf{Z}) \nrightarrow 0.$ 

The model that introduces an external eavesdropper to a TWC is called two-way wiretap channel (TW-WC). The TW-WC was first investigated in \cite{TY2007, TY2008} for both the Gaussian TW-WC and the binary additive TW-WC. 
Inner bounds on the weak secrecy capacity regions for both channels were derived using non-adaptive coding. 
The general discrete memoryless TW-WC was considered in \cite{GKYG2013}. A weak joint secrecy achievable region was established therein by using an instance of adaptive coding (i.e., each user sacrifices part of its secret rate to transmit a key to the other user and the keys are to be used in encrypting partial messages in next transmission round). %
Besides the weak joint secrecy, weak one-sided secrecy and individual secrecy were considered in \cite{QCHT2016} and \cite{QDT2017}, respectively, and respective weak secrecy rate regions were established therein.

Besides the above mentioned studies on TW-WCs under weak secrecy, \cite{PB2011} considered the general discrete memoryless TW-WC channel under a strong joint secrecy constraint, where the information leakage to the eavesdropper (rather than the leakage rate) is required to be negligible (i.e., $I(M_1, M_2;\mathbf{Z}) \to 0$). 
Strong secrecy for the TW-WC was studied using non-adaptive coding in \cite{HC2023} and adaptive coding in \cite{CH2025}, under joint or individual secrecy constraints. 

In this paper, we study the TW-WC under a strong mixed/asymmetric secrecy constraint, namely, strong one-sided secrecy in which only User~1's message is required to be secure. We consider two adaptive key-exchange constructions. The first is the full symmetric rate-splitting construction inherited from the multiround adaptive coding scheme of \cite{CH2025}. By specializing its multiround reliability and secrecy analysis to the leakage of User~1's payload, we obtain a strong one-sided achievable region for this construction. We then introduce a one-sided reduced construction that removes the key and encrypted-message components that are unnecessary for User~2. For the reduced construction, we provide a direct multiround analysis based on selected-subindex resolvability, a one-time-pad inequality with side information, and an ideal-key to actual-code transfer argument. The two operationally achievable constructions yield the same projected payload-rate region, showing that the reduced construction has a simpler design without loss in the achievable region. In particular, User~2's unprotected message contributes to the averaging that protects User~1. We derive explicit non-adaptive and adaptive achievable regions, identify the exact strict feasibility conditions for the auxiliary-rate projections, and compare the resulting construction-specific inner bounds. 

The remainder of this paper is organized as follows. In Section~\ref{sec:model}, we formulate the TW-WC under the strong mixed-secrecy constraint and introduce the necessary preliminaries. In Sections~\ref{sec: Non-Adaptive Coding} and~\ref{sec: Adaptive coding}, we present the non-adaptive and adaptive coding constructions, respectively, and derive their achievable strong mixed-secrecy regions. Section~\ref{Sec: Conclusion} concludes the paper. The appendices provide detailed proofs and Fourier--Motzkin elimination steps used to derive the stated regions.

\section{Channel model and Preliminaries}\Label{sec:model}
\subsection{Formulation}
We consider a discrete memoryless TW-WC, where 
two legitimate users, User 1 and User 2 intend to exchange 
messages with each other in the presence of an external eavesdropper.
The channel is characterized by $P_{Y_1,Y_2,Z|X_1,X_2}.$ 
The channel model is shown in Fig. \ref{fig: TW-WTC with an external eavesdropper}. 

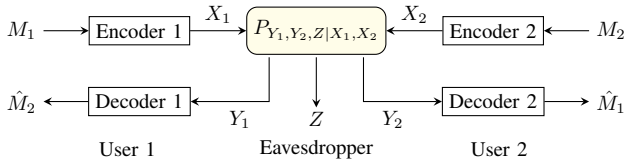
\begin{figure}[htbp]
	\centering
	\begin{tikzpicture}[>=stealth,scale=.78,transform shape]
		\node (w1e) at (0,0) {$M_{1}$};
		\node [draw] (en1) at (2,0) {Encoder 1};
		\node [draw,rounded corners,fill=yellow!10,minimum width=2cm,minimum height=0.80cm] (c) at (5,0) {$P_{Y_1,Y_2,Z|X_1,X_2}$};
		\node [draw] (en2) at (8,0) {Encoder 2};
		\node (w2) at (10,0) {$M_2$};	
		\node (y1) at (0,-1.2) {$\hat{M}_{2}$};
		\node [draw] (de1) at (2,-1.2) {Decoder 1};
		\node [draw] (de2) at (8,-1.2) {Decoder 2};
		\node (y2) at (10,-1.2) {$\hat{M}_1$};	
		\path[-stealth] (w1e) edge (en1);
		\draw[->] (en1) --  (c) node[midway, above=0.1] {$X_1$};
		\draw[->] (en2) --  (c) node[midway, above=0.1] {$X_2$};
		\path[-stealth] (w2) edge (en2);
		\path[-stealth] (de1) edge (y1);
		\path[-stealth] (de2) edge (y2);
		\draw[->] (4.2,-0.45) -- (4.2,-1.2) -- (de1);
		\draw[->] (5.8,-0.45) -- (5.8,-1.2) -- (de2);
		\draw[->] (5,-0.45) -- (5,-1.3);
		\node at (3.7,-1.5) {$Y_1$};
		\node at (6.3,-1.5) {$Y_2$};
		\node at (5,-1.5) {$Z$};
		\node at (1.8,-2) {User 1};
		\node at (8.1,-2) {User 2};
		\node at (5,-2) {Eavesdropper};
	\end{tikzpicture}
	\caption{Two-Way wiretap channel with an external eavesdropper.} \Label{fig: TW-WTC with an external eavesdropper}
\end{figure}

All logarithms are natural, and rates are measured in nats per channel use. The messages $M_i$ are assumed to be uniformly distributed over the message sets
$\mathcal{M}_i=[1:\lfloor e^{nR_i}\rfloor]$ for $i=1,2$. As usual, integer roundings of exponential codebook sizes are suppressed below because they do not affect the asymptotic rates.

Consider the communication in $n$ channel uses.
We denote User $i$'s channel input and output 
by $X_i^n=(X_{i,1}, \ldots, X_{i,n}) \in {\cal X}_i^n$ and 
$Y_i^n=(Y_{i,1}, \ldots, Y_{i,n})\in {\cal Y}_i^n$, respectively.
Also, we denote the channel output at the eavesdropper by 
$Z^n=(Z_{1}, \ldots, Z_{n})\in {\cal Z}^n$. 

We consider two types of encoders at two legitimate users.
\begin{itemize}
	\item The first type of encoder is a {\it non-adaptive} encoder $\phi_i$, which stochastically assigns the whole input $X_i^n$ based on the message $M_i$ for $i=1,2$.
	\item The second type of encoder is an {\it adaptive} encoder $\phi_i$, which stochastically assigns the $t$-th input $X_{i,t}$ based on the message $M_i$ and 
	the previous outputs $Y_{i,1},\ldots, Y_{i,t-1}$ for $t=1,\ldots,n$ and $i=1,2$.
\end{itemize}
For $i\in\{1,2\}$, let $i\oplus1:=3-i$. The decoder $\psi_i$ of User $i$ is a map from ${\cal M}_i\times{\cal X}_i^n\times{\cal Y}_i^n$ to ${\cal M}_{i \oplus 1}$; it may use the user's own message and transmitted sequence together with the received sequence. 

A $(e^{nR_1}, e^{nR_2}, n)$ secrecy code 
$\Co_n$ for the TW-WC consists of
$2$ message sets $\mathcal{M}_1, \mathcal{M}_2$,
$2$ encoders $\phi_1$, $\phi_2$, and
$2$ decoders $\psi_1$, $\psi_2$. Whether the code is adaptive or non-adaptive depends on whether adaptive or non-adaptive encoders are used. 

To evaluate the reliability of the transmission, we consider the \emph{average probability of decoding error} at the legitimate receiver that is defined by
	\begin{equation}\Label{eqn: Pe}
		P_{e}^n(\Co_n)=\frac{1}{|\mathcal M_1||\mathcal M_2|}\sum_{m_1, m_2}\Pr\left\{\bigcup\limits_{i\in \{1,2\}} \{m_i\neq \hat{m}_i\}\Bigg|\Co_n\right\}.
	\end{equation}

For secrecy, only User~1 requires protection from the eavesdropper. We call this requirement strong one-sided secrecy; it is the mixed secrecy constraint in the title. No secrecy constraint is imposed on $M_2$. 
We say that the rate pair $(R_{1,n},R_{2,n})$ is \emph{achievable under the strong mixed secrecy constraint by adaptive (non-adaptive) codes}, if there exists a sequence of $(e^{nR_{1,n}}, e^{nR_{2,n}}, n)$ adaptive (non-adaptive) codes $\{\Co_n\}$ such that
$R_{i,n} \to R_i$ for $i=1,2$, and the following bounds hold:
\begin{align} 
	P_{e}^n(\Co_n) &\leq \epsilon_n, \Label{eq:Reliability} \\
	I(M_1;Z^n|\Co_n)	&\leq	\tau_n, \Label{eq:SSec}
\end{align}
with $\lim\limits_{n\to\infty} \epsilon_n = 0$ and $\lim\limits_{n\to\infty} \tau_n= 0.$ 

\subsection{Preliminaries}\Label{sec: Renyi Lemmas}
In this section, we introduce some definitions that will be used in the paper.

First, we recall that the {\it R\'enyi relative entropy} is defined as follows:
\begin{align}\Label{eqn: Renyi relative entropy}
	D_{1+s}(P \| Q):=\frac{1}{s}\log  \sum_{x}P(x)^{1+s} Q(x)^{-s}.
\end{align}
Note that $D_{1+s}(P \| Q)$ is nondecreasing w.r.t. $s$ for $s>0$ and $\lim\limits_{s\to 0}D_{1+s}(P \| Q)=D(P \| Q),$ i.e., the relative entropy.

Following the notation in \cite[Eq. (36)]{HT16} and \cite[Eqs. (50), (52)]{TH}, we define the {\it R\'enyi mutual information} by the following expression:
\begin{align}
	I_{1+s}^{\uparrow} ( Z ; X) 
	:=D_{1+s} (P_{Z X}
	\| P_Z \times P_{X}   ).\Label{Ne1}
\end{align}
We define the {\it conditional R\'enyi mutual information} by the following identity:
\begin{align}
	& e^{-s I_{\frac{1}{1+s}}^{\downarrow}( Z ; X|Y) }
	:=\\
	&\quad
	\sum_y P_Y(y)
	e^{-s \min\limits_{Q_{Z|Y=y}}D_{\frac{1}{1+s}} (P_{Z X|Y=y} \|  Q_{Z|Y=y} \times P_{X|Y=y}   )}. \nonumber
\end{align}
The minimizing conditional distribution is
\begin{equation}
	Q_{Z|Y}^*(z|y)= \frac{\sum\limits_{ x} P_{X|Y}(x|y)
		P_{Z| X Y}(z|x,y)^{\frac{1}{1+s}}}{\sum\limits_{z'} \sum\limits_{ x'} P_{X|Y}(x'|y)
		P_{Z| X Y}(z'|x',y)^{\frac{1}{1+s}}},
\end{equation}
Substituting this minimizer gives the following expression for $I_{\frac{1}{1+s}}^{\downarrow}(Z;X|Y)$ \cite[Eq. (54)]{TH}:
\begin{align}
	&	e^{-s I_{\frac{1}{1+s}}^{\downarrow} ( Z ; X| Y ) }\Label{Ne3}\\
	=
	&\sum_{ y} P_Y(y)
	\sum_{z} 
	\Big(
	\sum_{ x} P_{X|Y}(x|y)
	P_{Z| X Y}(z|x,y)^{\frac{1}{1+s}}
	\Big)^{1+s}.\notag
\end{align}

Note that we have
\begin{align}
	\lim_{s\to 0} I_{1+s}^{\uparrow} (Z;X)& = I(Z;X); \label{lim_up}\\
	\lim_{s\to 0} I_{\frac{1}{1+s}}^{\downarrow} (Z;X|Y) &= I(Z;X|Y). \label{lim_down}
\end{align}

\section{Non-Adaptive Coding}\label{sec: Non-Adaptive Coding}

In this section, we consider the case where non-adaptive codes are used by both legitimate users.
Throughout this section, let $\mathcal Q$ denote the set of all probability distributions on $\mathcal V_1\times\mathcal V_2\times\mathcal X_1\times\mathcal X_2$ having the following factorization:
\begin{equation}\label{def: Q}
\mathcal Q:=\left\{P_{V_1}P_{V_2}P_{X_1\mid V_1}P_{X_2\mid V_2}\right\}.
\end{equation}
Thus every $P\in\mathcal Q$ satisfies
$P_{V_1V_2X_1X_2}=P_{V_1}P_{V_2}P_{X_1\mid V_1}P_{X_2\mid V_2}$.
Together with the fixed channel $P_{Y_1,Y_2,Z\mid X_1,X_2}$, each $P\in\mathcal Q$ induces the joint distribution under which the information quantities below are evaluated.

\subsection{Code construction}
Fix an arbitrary distribution of the form
\[
P=P_{V_1}P_{V_2}P_{X_1\mid V_1}P_{X_2\mid V_2}\in\mathcal Q.
\]
{\it Codebook Generation:}
At transmitter $i,$ let $V_{i,1}^n,\ldots, V_{i, \sM_i \cdot \sL_i}^n$
be random variables that are independently generated subject to
$P_{V_i^n}=\prod_{j=1}^{n}P_{V_{i,j}}, $ where
$V_i^n=(V_{i,1}, \cdots, V_{i,n}),$
$\sM_i= e^{n R_i}$ and
$\sL_i= e^{n R_{i,r}}$ for $i=1,2.$

\noindent {\it Encoding:}
To send the message $m_i$, the legitimate user $i$ chooses one of
${V}^n_{i,(m_i-1) \sL_i+1}, \ldots, {V}^n_{i,m_i \sL_i}$ with equal probability, and denotes the selected codeword by ${V}^n_{i,(m_i,m_{i,r})}$. Conditional on this codeword, transmitter $i$ generates $X_i^n$ memorylessly according to $\prod_{j=1}^n P_{X_i|V_i}(x_{i,j}|v_{i,j})$ and transmits the resulting $X_i^n$.

\noindent {\it Decoding:}
The other user applies ML decoding to ${V}^n_{i,(m_i,m_{i,r})}$ (and hence $m_i$) using its own transmitted sequence $X_{i\oplus1}^n$ and received sequence $Y_{i\oplus1}^n$.

\subsection{Exponential evaluation}
Specializing the User~1 leakage bound in \cite[Lemma~4]{HC2023}, we obtain the following one-sided specialization, which gives exponentially decreasing decoding error and information leakage.

For every $P\in\mathcal Q$ and $s\in(0,1]$, define the finite-$s$ quantities as follows:
\begin{align*}
 A_s(P)&:=I_{\frac{1}{1+s}}^{\downarrow}(Y_2;V_1\mid X_2),&
 B_s(P)&:=I_{\frac{1}{1+s}}^{\downarrow}(Y_1;V_2\mid X_1),\\
 C_s(P)&:=I_{1+s}^{\uparrow}(Z;V_1),&
 D_s(P)&:=I_{1+s}^{\uparrow}(Z;V_2),\\
 E_s(P)&:=I_{1+s}^{\uparrow}(Z;V_1,V_2),
\end{align*}
and define the corresponding Shannon-information quantities as follows:
\begin{align*}
 A(P)&:=I(Y_2;V_1\mid X_2),& B(P)&:=I(Y_1;V_2\mid X_1),\\
 C(P)&:=I(Z;V_1),& D(P)&:=I(Z;V_2),\\
 E(P)&:=I(V_1,V_2;Z).
\end{align*}
Here all quantities on the right-hand sides are evaluated under the joint distribution induced by $P$ and the fixed channel. By \eqref{lim_up}--\eqref{lim_down}, the following convergence holds for every fixed $P\in\mathcal Q$ as $s\downarrow0$:
\begin{align}
&\bigl(A_s(P),B_s(P),C_s(P),D_s(P),E_s(P)\bigr) \notag\\
&\longrightarrow
\bigl(A(P),B(P),C(P),D(P),E(P)\bigr)
\end{align}
Throughout the paper, the dependence of these quantities on $P$ is displayed explicitly.

\begin{lemma}\Label{L3T}
For every fixed $P\in\mathcal Q$, every fixed $s\in(0,1]$, a rate pair $(R_1,R_2)$, and two positive numbers $R_{1,r},R_{2,r}$, there exists a sequence of $(e^{nR_1}, e^{nR_2}, n)$ non-adaptive codes $\Co_n$ satisfying the following reliability and leakage bounds:
\begin{align}
P_{e}^n(\Co_n)
\leq &2 \Big[e^{ns\left(R_1+R_{1,r}-A_s(P)\right)}
+e^{ns\left(R_2+R_{2,r}-B_s(P)\right)}\Big]
; \Label{eq:Reliability2D} \\
I(M_1;Z^n|\Co_n)
\leq &2 \Big[e^{ns\left(E_s(P)-(R_{1,r}+R_{2,r}+R_2)\right)}
\notag\\
&+e^{ns\left(C_s(P)-R_{1,r}\right)}+e^{ns\left(D_s(P)-(R_{2,r}+R_2)\right)}\Big].
\Label{eq:SSec2D}
\end{align}
\end{lemma}
\begin{IEEEproof}
This is the User~1 part of \cite[Lemma~4]{HC2023}. The expectation bounds in \cite[(50)]{HC2023} and \cite[(52)]{HC2023}, followed by Markov's inequality, yield one code satisfying both \eqref{eq:Reliability2D} and \eqref{eq:SSec2D}.
\end{IEEEproof}

\subsection{Achievable secrecy region}
Define the strict feasibility set
\begin{align}
\mathcal Q_{\rm F}:=
\Bigl\{P\in\mathcal Q:\;&
C(P)<A(P),\notag\\
&D(P)<B(P),\notag\\
&E(P)<A(P)+B(P)
\Bigr\}.
\label{eq:QN-definition}
\end{align}
These conditions characterize the nonemptiness of the strict
auxiliary-rate system used below.

For each $P\in\mathcal Q_{\rm F}$, define the fixed-distribution
non-adaptive region
\begin{align}
\mathcal R_{\rm N}(P):=
\Bigl\{(R_1,R_2)\in\mathbb R_+^2:\;&
R_1\leq A(P)-C(P),\notag\\
&R_1\leq A(P)+B(P)-E(P),\notag\\
&R_2\leq B(P)
\Bigr\}.
\label{HH2NN}
\end{align}

\begin{lemma}\label{lem:NA-achievable-region}
The following payload-rate region is achievable by non-adaptive
codes under the strong mixed-secrecy criterion:
\begin{equation}
\mathcal R_{\rm N}
:=
\overline{\operatorname{conv}}\!\left(
\bigcup_{P\in\mathcal Q_{\rm F}}
\mathcal R_{\rm N}(P)
\right).
\label{eq:NA-achievable-region}
\end{equation}
\end{lemma}
\begin{IEEEproof}
\textit{Step 1: Exact projection for fixed $P$.}
Fix $P\in\mathcal Q_{\rm F}$. Consider nonnegative auxiliary
rates $R_{1,r}$ and $R_{2,r}$ satisfying the following inequalities:
\begin{align}
R_2+R_{2,r}&<B(P), \Label{G5A}\\
R_1+R_{1,r}&<A(P), \Label{G4A}\\
R_{1,r}&>C(P), \Label{G1}\\
R_2+R_{2,r}&>D(P), \Label{G2}\\
R_{1,r}+R_2+R_{2,r}&>E(P). \Label{G3}
\end{align}
Introduce the aggregate rates $x$ and $y$ by
\begin{equation}
x:=R_{1,r},
\qquad
y:=R_2+R_{2,r}.
\end{equation}
The auxiliary-rate constraints are equivalent to the following system of inequalities:
\begin{align}
C(P)&<x<A(P)-R_1,\notag\\
\max\{D(P),R_2\}&<y<B(P),\notag\\
x+y&>E(P).
\label{eq:NA-xy-system}
\end{align}
For this fixed $P$, such auxiliary rates exist if and only if the following three inequalities hold:
\begin{align}
R_1&<A(P)-C(P),\notag\\
R_1&<A(P)+B(P)-E(P),\notag\\
R_2&<B(P).
\label{eq:NA-fixed-P-strict-region}
\end{align}
Necessity follows directly from \eqref{eq:NA-xy-system}.
Conversely, the strict inequalities in
\eqref{eq:NA-fixed-P-strict-region}, together with
$P\in\mathcal Q_{\rm F}$, allow $x$ and $y$ to be chosen so that
all inequalities in \eqref{eq:NA-xy-system} are strict.
Consequently, the closure of the projection onto the
$(R_1,R_2)$-coordinates is $\mathcal R_{\rm N}(P)$.

\textit{Step 2: Finite-$s$ achievability for fixed $P$.}
Consider a rate pair in the interior of
$\mathcal R_{\rm N}(P)$. By Step~1, the auxiliary rates can be
chosen so that \eqref{G5A}--\eqref{G3} hold with a common
strictly positive slack. Since $P$ is fixed, the continuity
relations in the preceding subsection imply that there exists a
sufficiently small fixed $s\in(0,1]$ such that all of the following inequalities hold:
\begin{align}
R_1+R_{1,r}&<A_s(P),\notag\\
R_2+R_{2,r}&<B_s(P),\notag\\
R_{1,r}&>C_s(P),\notag\\
R_2+R_{2,r}&>D_s(P),\notag\\
R_{1,r}+R_2+R_{2,r}&>E_s(P).
\label{eq:NA-finite-s-system}
\end{align}
Lemma~\ref{L3T} then gives exponentially decreasing decoding
error and information leakage. Hence every interior point of
$\mathcal R_{\rm N}(P)$ is achievable. Its boundary points are
obtained by choosing a sequence of achievable interior rate pairs
converging to the desired point.

\textit{Step 3: Union, time sharing, and closure.}
The code distribution may be chosen arbitrarily from
$\mathcal Q_{\rm F}$. Therefore, every rate pair in
\begin{equation}
\bigcup_{P\in\mathcal Q_{\rm F}}
\mathcal R_{\rm N}(P)
\end{equation}
is achievable. Time sharing among finitely many non-adaptive codes
remains non-adaptive and gives the convex hull of this union.
Finally, a standard diagonal argument gives its closure. Thus every
rate pair in $\mathcal R_{\rm N}$ is achievable.
\end{IEEEproof}

\section{Adaptive coding}\label{sec: Adaptive coding}
In this section, both legitimate users employ adaptive key exchange. We first describe the full symmetric construction to explain its relation to earlier schemes, and then give the one-sided reduced construction. 
The adaptive achievable region is proved using the reduced construction.
We use $t$ for the round index and $T$ for the total number of rounds.  

\begin{figure*}[!ht]
	\centering	
	\begin{tabular}{rcl}		
		$V_1^n:$ &	& 
		$
		\overbrace{
			\overbrace{	
				\begin{tikzpicture}
					\node[minimum height=1.6em, minimum width=10em, anchor=base, fill=blue!25] {$m_{1,\mathrm{sec}}^t$}; 
				\end{tikzpicture}
			}^{nR_{1,\mathrm{sec}}}
			\overbrace{	
				\begin{tikzpicture}
					\node[minimum height=1.6em, minimum width=8em, anchor=base, fill=yellow!25] {$m_{1,k}^t$}; 
				\end{tikzpicture}
			}^{nR_{1,k}}
		}^{n\tilde{R}_1}
		\overbrace{
			\overbrace{
				\begin{tikzpicture}
					\node[minimum height=1.6em, minimum width=2em, anchor=base, fill=red!25] {$m_{1,e}^t\oplus \hat{m}_{2,k}^{t-1}$}; 
				\end{tikzpicture}
			}^{nR_{1,e}}		
			\overbrace{
				\begin{tikzpicture}
					\node[minimum height=1.6em, minimum width=1.5em, anchor=base, fill=red!25] {$m_{1,o}^t$}; 
				\end{tikzpicture}
			}^{nR_{1,o}}
		}^{n\tilde{R}_{1,r}}
		$ \\
		$V_2^n:$ & 	& 
		$
		\underbrace{
			\underbrace{	
				\begin{tikzpicture}
					\node[minimum height=1.6em, minimum width=10em, anchor=base, fill=teal!25] {$m_{2,\mathrm{sec}}^t$}; 
				\end{tikzpicture}
			}_{nR_{2,\mathrm{sec}}}
			\underbrace{	
				\begin{tikzpicture}
					\node[minimum height=1.6em, minimum width=8em, anchor=base, fill=yellow!25] {$m_{2,k}^t$}; 
				\end{tikzpicture}
			}_{nR_{2,k}}
		}_{n\tilde{R}_{2}}
		\underbrace{
			\underbrace{
				\begin{tikzpicture}
					\node[minimum height=1.6em, minimum width=2.5em, anchor=base, fill=red!25] {${m}_{2,e}^t\oplus \hat{m}_{1,k}^{t-1}$}; 
				\end{tikzpicture}
			}_{nR_{2,e}}		
			\underbrace{
				\begin{tikzpicture}
					\node[minimum height=1.6em, minimum width=1.5em, anchor=base, fill=red!25] {$m_{2,o}^t$}; 
				\end{tikzpicture}
			}_{nR_{2,o}}
		}_{n\tilde{R}_{2,r}}
		$ 
	\end{tabular}
	\caption{Encoding for adaptive coding in round $t$.}
	\label{fig: Message structure b}	
\end{figure*}
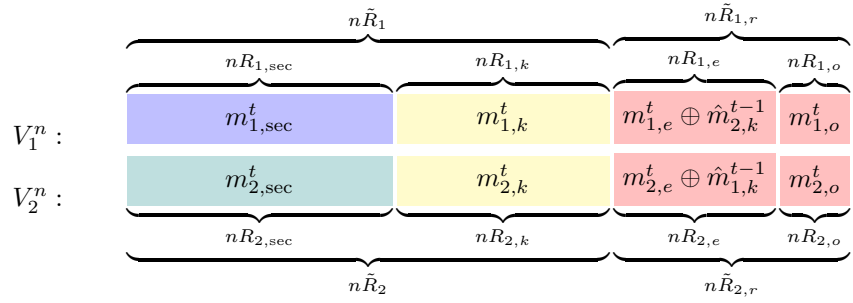
\subsection{Code construction I}

We use the subscript ``sec'' for the directly protected secret-message component, reserving the symbol $s$ exclusively for the R\'enyi parameter.

The idea of adaptive code as presented in \cite{CH2025} is to operate the non-adaptive coding for several rounds, and in each round the message at each user is split into several parts, which include not only the secret and public parts, but also a key part and an encrypted message part. In more detail, at round $t,$ to send $m_i^t=(m_{i,\mathrm{sec}}^t, m_{i, e}^t),$ User $i$ randomly chooses $m_{i,r}^t = (m_{i,k}^t, m_{i,o}^t),$ where   
\begin{enumerate}
	\item $\tilde{m}_i^t=(m_{i,\mathrm{sec}}^t, m_{i,k}^t)$ are the message parts to be kept secret (at round $t$) that include
	\begin{itemize}
		\item a secret message $m_{i,\mathrm{sec}}^t,$ where $m_{i,\mathrm{sec}}^t\in [1, e^{nR_{i,\mathrm{sec}}}],$
		\item a key $m_{i,k}^t,$ which is to be used for encryption by the other user for the next round, where $m_{i, k}^t\in [1, e^{nR_{i,k}}],$
	\end{itemize}
	\item  $\tilde{m}_{i,r}^t=(m_{i, e}^t\oplus \hat{m}_{i\oplus 1,k}^{t-1}, m_{i,o}^t)$ are the message parts that are not required to be secret that include
	\begin{itemize}
		\item an encrypted message $m_{i, e}^t\oplus \hat{m}_{i\oplus 1, k}^{t-1}$, where the message part $m_{i, e}^t$ is encrypted using one-time pad with the key from the other user from the previous round $\hat{m}_{i\oplus 1,k}^{t-1}$, where $m_{i, e}^t\in [1, e^{nR_{i,e}}]$ and  $\hat{m}_{i\oplus 1,k}^{t-1}\in [1, e^{nR_{i\oplus 1,k}}].$ The encrypted component is protected by a one-time pad using the key generated by the other user in the preceding round, which requires $R_{i,e}\leq R_{i\oplus 1,k}.$ Note that this key may be correlated with Eve's previous observations, as shown in the dependency graph in \cite[Fig. 4]{CH2025}. 
		\item an open message $m_{i,o}^t,$ where $m_{i, o}^t\in [1, e^{nR_{i,o}}].$
	\end{itemize}
\end{enumerate}  
See Fig. \ref{fig: Message structure b}	for an illustration of the code construction.

{\it Codebook Generation:}
At transmitter $i,$ let $V_{i,1}^n,\ldots, V_{i, \sM_i \cdot \sL_i}^n$ 
be random variables independently subject to 
$P_{V_i^n}=\prod_{j=1}^{n}P_{V_{i,j}}, $ where 
$V_i^n=(V_{i,1}, \cdots, V_{i,n}),$  
$\sM_i= e^{n \tilde{R}_i}$ and 
$\sL_i= e^{n \tilde{R}_{i,r}}$ for $i=1,2.$

{\it Encoding:}
At round $t>1,$ when the legitimate user $i$ intends to send the message $m_i^t=(m_{i,\mathrm{sec}}^t, m_{i, e}^t)$, the user randomly chooses $m_{i,k}^t\in [1, e^{nR_{i,k}}]$ and  $m_{i, o}^t\in [1, e^{nR_{i,o}}]$ and sends ${V}^n_{i,(\tilde{m}_i^t, \tilde{m}_{i,r}^t)},$ with $\tilde{m}_i^t=(m_{i,\mathrm{sec}}^t, m_{i,k}^t)$ and $\tilde{m}_{i,r}^t=(m_{i, e}^t\oplus \hat{m}_{i\oplus 1,k}^{t-1}, m_{i,o}^t)$. Here $\hat{m}_{i\oplus 1,k}^{t-1}$ is User~$i$'s decoded estimate of the key generated by the other user in the preceding round. The actual encoder uses this decoded estimate; hats are suppressed only in rate accounting, because the true and decoded key alphabets have the same size. 

The first round initializes the key exchange and carries no actual payload. For each user $i$, let $D_{i,\mathrm{sec}}$ and $D_{i,e}$ be independent uniform dummy coordinates of rates $R_{i,\mathrm{sec}}$ and $R_{i,e}$, respectively. User~$i$ also generates an independent uniform key $m_{i,k}^1$ and an independent uniform open coordinate $D_{i,o}$ of rate $R_{i,o}$. The round-1 indices are
\[
 \tilde m_i^1=(D_{i,\mathrm{sec}},m_{i,k}^1),\qquad
 \tilde m_{i,r}^1=(D_{i,e},D_{i,o}).
\]
Thus the nominal index dimensions and averaging coordinates in round~1 agree with those in rounds $2,\ldots,T$. Both receivers decode all dummy and key coordinates, every such decoding error is included in the block-error event, and the dummy coordinates are subsequently discarded.

{\it Decoding:}
At the other legitimate receiver, an ML decoder will be used to decode ${V}^n_{i,(\tilde{m}_i^t, \tilde{m}_{i,r}^t)}$ (and therefore obtain an estimate of $(\tilde{m}_i^t, \tilde{m}_{i,r}^t)$) by using the receiver's information $X_{i \oplus 1}^n$. Thus $\hat{m}_i^t=(\hat{m}_{i,\mathrm{sec}}^t, \hat{m}_{i, e}^t)$ is obtained by taking $\hat{m}_{i,\mathrm{sec}}^t$ from the estimate of $\tilde{m}_i^t$ and taking $\hat{m}_{i,e}^t$ from decoding the estimate of $\tilde{m}_{i,r}^t$ with $m_{i\oplus 1,k}^{t-1}.$ Note that $m_{i\oplus 1,k}^{t-1}$ is the key coordinate generated in the previous round by the receiver and intended to be protected from the eavesdropper; its quantitative leakage is not inferred from the alphabet-size condition alone. $\hat{m}_{i,k}^{t}$ from the estimate of $\tilde{m}_i^t$ will be decoded as well (and used in the next round encoding).

In this way, it is possible for User $i$ to send the message $m_i$ (including the secret message part and encrypted message part) of rate $R_i$ (except the first round), where the payload rates are given by
\begin{align}
	R_1 =& R_{1,\mathrm{sec}} + R_{1,e}, \Label{G8A-full}\\ 
	R_2 =& R_{2,\mathrm{sec}} + R_{2,e}. \Label{G9A-full}
\end{align}  
using an underlying $(e^{n\tilde{R}_1}, e^{n\tilde{R}_2}, n)$ non-adaptive code, whose rates are given by
\begin{align}
	\tilde{R}_i =& R_{i,\mathrm{sec}} + R_{i,k}, \Label{G11A}\\
	\tilde{R}_{i,r} =& R_{i,o} + R_{i,e}. \Label{G12A}
\end{align}
Note that $R_{i,\mathrm{sec}}$ is the rate of the secret message part; $R_{i,k}$ is the rate of the key part; $R_{i,o}$ is the rate of the partial messages that could be public; and $R_{i,e}$ is the rate of encrypted message part. 
The following rate-matching conditions ensure that each encrypted component can be embedded into the available key alphabet: 
\begin{align}
	R_{1,e} \leq & R_{2,k}, \Label{G6A-full}\\
	R_{2,e} \leq & R_{1,k}. \Label{G7A-full}
\end{align}


Fix an integer $T\geq2$. Round~1 is an initialization round and carries no actual payload, whereas rounds $2,\ldots,T$ carry payload at rates $(R_1,R_2)$. Thus the total blocklength is $nT$, the payload sizes are $e^{n(T-1)R_i}$, and the effective rate of User~$i$ is $(T-1)R_i/T$, which converges to $R_i$ as $T\to\infty$. The alphabet-size conditions alone do not imply perfect secrecy, because a key generated in a preceding round can be correlated with Eve's past observations; the multiround analysis below accounts for this dependence. 

%

\subsection{Exponential evaluation for code construction I}
The full construction inherits a multiround reliability and secrecy analysis from \cite{CH2025}. 
For comparison with the reduced construction, we state the following one-sided specialization. The proof of the adaptive achievable-region theorem below is based instead on the reduced construction.
It is stated in terms of the actual payload transmitted in rounds $2,\ldots,T$.
\begin{lemma}[Full-construction multiround bound]\Label{L3TA-FA}
For every fixed $P\in\mathcal Q$, every fixed $s\in(0,1]$, an integer $T\geq2$, and nonnegative splitting rates satisfying the rate-matching conditions \eqref{G6A-full}--\eqref{G7A-full}, there exists a full adaptive code $\Co_n^T$ of blocklength $nT$ and payload sizes $e^{n(T-1)R_1}$ and $e^{n(T-1)R_2}$, where the payload rates are
\[
 R_1=R_{1,\mathrm{sec}}+R_{1,e},\qquad R_2=R_{2,\mathrm{sec}}+R_{2,e},
\]
and the underlying code rates are
\[
 \widetilde R_i=R_{i,\mathrm{sec}}+R_{i,k},\qquad \widetilde R_{i,r}=R_{i,e}+R_{i,o},
\]
The code satisfies the following reliability and leakage bounds:
\begin{align}
 P_e^{nT}(\Co_n^T)
 &\leq2T\Big[
 e^{ns(\widetilde R_2+\widetilde R_{2,r}-B_s(P))}\notag\\
 &\hspace{14mm}+
 e^{ns(\widetilde R_1+\widetilde R_{1,r}-A_s(P))}
 \Big], \Label{eq:ReliabilityFA}\\
 I(\mathbf M_1^{2:T};Z^{nT}\mid\Co_n^T)
 &\leq2T\Big[
 e^{ns(E_s(P)-\widetilde R_{1,r}-\widetilde R_{2,r}-R_{2,\mathrm{sec}})}\notag\\
 &\hspace{14mm}+e^{ns(C_s(P)-\widetilde R_{1,r})}\notag\\
 &\hspace{14mm}+e^{ns(D_s(P)-\widetilde R_{2,r}-R_{2,\mathrm{sec}})}
 \Big]. \Label{eq:SSecFA}
\end{align}
Here $\mathbf M_1^{2:T}=(M_1^2,\ldots,M_1^T)$, where $M_1^t=(M_{1,\mathrm{sec}}^t,M_{1,e}^t)$, and $Z^{nT}=(Z^n[1],\ldots,Z^n[T])$.
\end{lemma}
\begin{IEEEproof}
The claim follows by specializing the multiround reliability and secrecy analysis of \cite[Secs.~5.3--5.4]{CH2025} to the leakage of User~1's payload. Under this specialization, User~1's secret and encrypted payload components are retained in the leakage criterion, whereas the message coordinates of User~2 that are not required to be secret contribute to the averaging in the resolvability analysis. The analysis in \cite{CH2025} treats the dependence between preceding-round keys and Eve's past observations and the use of decoded preceding-round keys by the actual encoders. After dropping the secrecy requirement on User~2's payload and using data processing to discard the round-1 dummy coordinates, its ensemble bounds reduce to \eqref{eq:ReliabilityFA} and \eqref{eq:SSecFA}. Applying Markov's inequality to the sum of the normalized error and leakage quantities yields one deterministic code satisfying both displayed bounds.
\end{IEEEproof}

\subsection{Code construction II: one-sided reduced construction}
Although code construction~I is operationally achievable by the multiround analysis inherited from \cite{CH2025}, its symmetric rate splitting contains components that are unnecessary under one-sided secrecy. Since no secrecy constraint is imposed on User~2, we impose
\begin{equation}\label{eq:one-sided-reduction}
	R_{2,e}=R_{1,k}=0,\qquad R_{2,\mathrm{sec}}=R_2.
\end{equation}
We directly analyze the resulting reduced construction. This gives a self-contained leakage recursion tailored to one-sided secrecy and explicitly shows how User~2's unprotected payload contributes to the averaging that protects User~1. At the level of the auxiliary-rate systems, moving User~2's encrypted component into its unprotected main component preserves both its payload rate and its contribution to the averaging, while eliminating the key that would otherwise be generated by User~1. Appendices~\ref{APP: FM-Full-Adaptive} and~\ref{APP: FM-OS-Adaptive} verify that the full and reduced systems have the same projected region.
See Fig. \ref{fig: Message structure c}	for an illustration of this code construction.
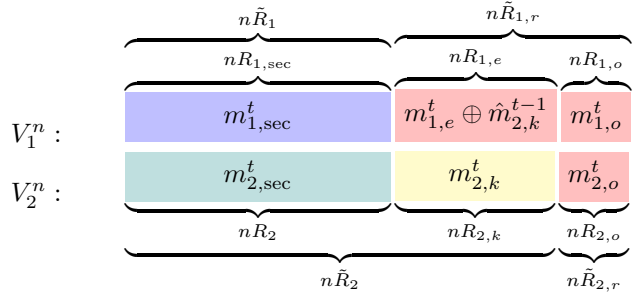
\begin{figure*}[h]
	\centering	
	\begin{tabular}{rcl}		
		$V_1^n:$ &	& 
		$
		\overbrace{
			\overbrace{	
				\begin{tikzpicture}
					\node[minimum height=1.6em, minimum width=10em, anchor=base, fill=blue!25] {$m_{1,\mathrm{sec}}^t$}; 
				\end{tikzpicture}
			}^{nR_{1,\mathrm{sec}}}
		}^{n\tilde{R}_1}
		\overbrace{
			\overbrace{
				\begin{tikzpicture}
					\node[minimum height=1.6em, minimum width=2em, anchor=base, fill=red!25] {$m_{1,e}^t\oplus \hat{m}_{2,k}^{t-1}$}; 
				\end{tikzpicture}
			}^{nR_{1,e}}		
			\overbrace{
				\begin{tikzpicture}
					\node[minimum height=1.6em, minimum width=1.5em, anchor=base, fill=red!25] {$m_{1,o}^t$}; 
				\end{tikzpicture}
			}^{nR_{1,o}}
		}^{n\tilde{R}_{1,r}}
		$ \\
		$V_2^n:$ & 	& 
		$
		\underbrace{
			\underbrace{	
				\begin{tikzpicture}
					\node[minimum height=1.6em, minimum width=10em, anchor=base, fill=teal!25] {$m_{2,\mathrm{sec}}^t$}; 
				\end{tikzpicture}
			}_{nR_{2}}
			\underbrace{	
				\begin{tikzpicture}
					\node[minimum height=1.6em, minimum width=6em, anchor=base, fill=yellow!25] {$m_{2,k}^t$}; 
				\end{tikzpicture}
			}_{nR_{2,k}}
		}_{n\tilde{R}_{2}}
		\underbrace{	
			\underbrace{
				\begin{tikzpicture}
					\node[minimum height=1.6em, minimum width=1.5em, anchor=base, fill=red!25] {$m_{2,o}^t$}; 
				\end{tikzpicture}
			}_{nR_{2,o}}
		}_{n\tilde{R}_{2,r}}
		$ 
	\end{tabular}
	\caption{Encoding for the one-sided reduced adaptive construction in round $t$.}
	\label{fig: Message structure c}	
\end{figure*}

An independent codebook is generated for every round. In each round, the nominal codeword indices are encoded by the underlying non-adaptive encoder, including its memoryless stochastic prefix channel $P_{X_i|V_i}$. 
More specifically, at round $t\geq2$, User~1 splits its payload into a directly protected component $M_{1,\mathrm{sec}}^t$ of rate $R_{1,\mathrm{sec}}$ and an encrypted component $M_{1,e}^t$ of rate $R_{1,e}$. To make the finite alphabets precise, choose finite abelian groups $G_n$ and $H_n$ such that $n^{-1}\log|G_n|\to R_{1,e}$ and $n^{-1}\log|G_n\times H_n|\to R_{2,k}$. User~2 sends its entire payload $M_{2,t}$ at rate $R_2$ and a fresh uniform key $K_t^{\rm full}$ on $G_n\times H_n$. Let $K_t$ be the first component of $K_t^{\rm full}$, i.e., the projection onto $G_n$, and require
\begin{equation}\label{G6A}
 R_{1,e}\leq R_{2,k}.
\end{equation}
The rate condition matches the encrypted component to an available key alphabet; secrecy in the presence of Eve's side information is quantified below rather than asserted to be perfect. Let $\widehat K_{t-1}$ denote User~1's estimate of the key generated by User~2 in round $t-1$. The actual adaptive encoder forms
\begin{equation}
C_t^{\rm real}:=M_{1,e}^t\oplus \widehat K_{t-1}.
\end{equation}
For the leakage recursion we use the coupled ideal-key process $C_t:=M_{1,e}^t\oplus K_{t-1}$. The two processes coincide whenever all preceding key indices have been decoded correctly. Lemma~\ref{lem:ideal-real-transfer} transfers the ideal-process leakage bound to the actual adaptive code.

To apply the underlying non-adaptive encoder in round $t$, each user combines its message, ciphertext, key, and open-randomization variables into two encoder indices. User~1 supplies $(\widetilde M_{1}^t,\widetilde M_{1,r}^t)$, and User~2 supplies $(\widetilde M_{2}^t,\widetilde M_{2,r}^t)$, where the independent open indices $M_{1,o}^t$ and $M_{2,o}^t$ have rates $R_{1,o}$ and $R_{2,o}$, respectively:
\begin{align}
 \widetilde M_{1}^t&=M_{1,\mathrm{sec}}^t,&
 \widetilde M_{1,r}^t&=(C_t^{\rm real}, M_{1,o}^t),\notag\\
 \widetilde M_{2}^t&=(M_{2,t},K_t^{\rm full}),&
 \widetilde M_{2,r}^t&=M_{2,o}^t.\label{eq:reduced-indices}
\end{align}
Here $\widetilde M_{i}^t$ is the main codebook index and $\widetilde M_{i,r}^t$ is the randomization index for User~$i$. Their rates satisfy the following identities:
\begin{align}
 R_1&=R_{1,\mathrm{sec}}+R_{1,e},& R_2&=R_{2,\mathrm{sec}},\label{G8A}\\
 \widetilde R_1&=R_{1,\mathrm{sec}},&
 \widetilde R_{1,r}&=R_{1,e}+R_{1,o},\notag\\
 \widetilde R_2&=R_2+R_{2,k},&
 \widetilde R_{2,r}&=R_{2,o}.\label{eq:reduced-rates}
\end{align}
Each receiver applies the same maximum-likelihood decoder as in the underlying non-adaptive code and recovers all nominal indices of the other user. User~2 decrypts $C_t^{\rm real}$ using its previously generated key $K_{t-1}$; on the event of correct preceding-round key decoding, $C_t^{\rm real}=C_t$.

The first round initializes the key exchange and carries no payload. User~1 generates independent uniform dummy coordinates $D_{1,\mathrm{sec}}$ and $D_{1,e}$ of rates $R_{1,\mathrm{sec}}$ and $R_{1,e}$, together with the independent open index $O_{1,1}$ of rate $R_{1,o}$. User~2 generates an independent uniform dummy coordinate $D_2$ of rate $R_2$, the initial full key $K_1^{\rm full}$, and the independent open index $O_{2,1}$ of rate $R_{2,o}$. The round-1 encoder indices are $(D_{1,\mathrm{sec}},(D_{1,e},O_{1,1}))$ and $((D_2,K_1^{\rm full}),O_{2,1})$, so their nominal dimensions agree exactly with \eqref{eq:reduced-indices}. Both receivers decode all dummy, key, and open coordinates, and every such decoding error is included in the block-error event. The dummy coordinates are subsequently discarded; data processing then leaves only the leakage of $K_1^{\rm full}$ needed to start the recursion.

For $T\geq2$ rounds, the total blocklength is $nT$, whereas rounds $2,\ldots,T$ carry payload. Hence the payload sizes are $e^{n(T-1)R_i}$ and the effective rates are $(T-1)R_i/T$.

\subsection{Exponential evaluation for code construction II}
We first isolate the one-block statement needed for the reduced construction. It is the selected-subindex specialization of the individual-leakage argument in \cite[Lemma~4 and Sec.~III-E]{HC2023}.

\begin{lemma}[Selected-subindex resolvability]\label{lem:selected-resolvability}
Fix $P\in\mathcal Q$. For one block, write User~1's retained main index as $A_1$, its averaged auxiliary index as $B_1$, User~2's retained key index as $A_2$, its unprotected main coordinate as $J_2$, and its averaged auxiliary index as $B_2$. Suppose that these coordinates are mutually independent and uniform before they are mapped to codewords, and that the current codebook is generated independently of them and of any variables from preceding blocks. Then, for every fixed $s\in(0,1]$,
\begin{align}
&\mathbb E_{\mathsf C}
 I(A_1,A_2;Z^n\mid\mathsf C)\notag\\
&\quad\leq
 e^{ns(E_s(P)-R_{B_1}-R_{J_2}-R_{B_2})}\notag\\
&\qquad+e^{ns(C_s(P)-R_{B_1})}
+e^{ns(D_s(P)-R_{J_2}-R_{B_2})}.
\label{eq:selected-resolvability}
\end{align}
Here ``retained'' means that the index appears on the left-hand side of the resolvability leakage bound, whether or not it is an ultimate payload. In particular, $A_2$ is the key temporarily tracked for use in the next round, whereas the unprotected payload coordinate $J_2$ contributes to averaging. An averaged index is mixed over in the resolvability argument.
\end{lemma}
\begin{IEEEproof}
The proof is given in Appendix~\ref{APP: Proof of Lemma Selected-Resolvability}.	
\end{IEEEproof}

\begin{lemma}[One-time pad with side information]\label{lem:otp-side-information}
Let $U$ and $K$ be uniform on the same finite abelian group $G$,
and assume that $U$ is independent of $(K,W,S)$ and that $K$ is
independent of $W$. Define the ciphertext by $C:=U\oplus K$. Then
\begin{equation}\label{eq:otp-side-information}
I(U;C,S\mid W)\leq I(K;S\mid W).
\end{equation}
If a larger uniform key $K^{\rm full}$ is available, the statement
remains valid with $K=\pi(K^{\rm full})$ for any fixed surjective
homomorphism $\pi$ onto $G$, and
$I(K;S\mid W)\leq I(K^{\rm full};S\mid W)$.
\end{lemma}
\begin{IEEEproof}
The proof is given in Appendix~\ref{APP: Proof of Lemma Otp-side-information}.
\end{IEEEproof}

\begin{lemma}[Ideal-key to actual-code transfer]\label{lem:ideal-real-transfer}
Fix $P\in\mathcal Q$ and a codebook collection $\boldsymbol{\mathsf C}=\boldsymbol c$, let $F_{n,T}$ be the event that at least one key used by an encoder in rounds $2,\ldots,T$ was decoded incorrectly in the preceding round, and define
\[
p_{n,T}(\boldsymbol c):=\Pr\{F_{n,T}\mid\boldsymbol{\mathsf C}=\boldsymbol c\}.
\] Couple the actual adaptive code, which uses $\widehat K_{t-1}$, and the ideal-key process, which uses $K_{t-1}$, with the same messages, keys, codebooks, and channel randomness until their first discrepancy. With total variation defined as one half of the $\ell_1$ distance, the two conditional laws of $(\mathbf M_1^{2:T},Z^{nT})$ are then at total variation distance at most $p_{n,T}(\boldsymbol c)$.

Let $\bar p_{n,T}:=\mathbb E_{\boldsymbol{\mathsf C}}p_{n,T}(\boldsymbol{\mathsf C})$ and assume $\bar p_{n,T}\leq1/2$. Then
\begin{align}
 \mathbb E_{\boldsymbol{\mathsf C}}
 I_{\rm real}(\mathbf M_1^{2:T};Z^{nT}\mid\boldsymbol{\mathsf C}) &\leq
 \mathbb E_{\boldsymbol{\mathsf C}}
 I_{\rm ideal}(\mathbf M_1^{2:T};Z^{nT}\mid\boldsymbol{\mathsf C})
 +\eta_{n,T},\label{eq:ideal-real-transfer}\\
 \eta_{n,T}&:=
 4\bar p_{n,T}\log|\mathcal{M}_1^{T-1}|
 +2h_2(\bar p_{n,T}).
 \label{eq:eta-definition}
\end{align}
where $\mathbf M_1^{2:T}:=(M_1^2,\ldots,M_1^T)$ and $Z^{nT}:=(Z^n[1],\ldots,Z^n[T])$. The actual and ideal processes have the same uniform marginal distribution on $\mathbf M_1^{2:T}$.
Moreover, the union bound and the one-block decoding estimate give
\begin{align}
\bar p_{n,T}
\leq T\Big[e^{ns(\widetilde R_2+\widetilde R_{2,r}-B_s(P))}
+e^{ns(\widetilde R_1+\widetilde R_{1,r}-A_s(P))}
\Big].
\label{eq:pbar-bound}
\end{align}
Thus $\eta_{n,T}$ decreases exponentially in $n$ for fixed $T$ under the strict reliability inequalities.
\end{lemma}
\begin{IEEEproof}
The proof is given in Appendix~\ref{APP: Proof of Lemma Ideal-real-transfer}.
\end{IEEEproof}

The following factorization records precisely why the one-block lemma can be applied conditionally in each round of the ideal process.
\begin{lemma}[Current-round index factorization]\label{lem:current-round-index-factorization}
Fix $t\geq2$ and condition on the codebooks outside round $t$. In the ideal process, let
\[
\mathcal H_{t-1}:=\sigma\!\left(\mathbf M_1^{2:t-1},K_{t-1}^{\rm full},Z^n[1:t-1],\boldsymbol{\mathsf C}_{<t},\boldsymbol{\mathsf C}_{>t}\right).
\]
The round-$t$ codebook $\mathsf C_t$ is independent of $\mathcal H_{t-1}$. Moreover, conditional on $\mathcal H_{t-1}$, the fresh variables $M_{1,\mathrm{sec}}^t$, $M_{1,e}^t$, $M_2^t$, $K_t^{\rm full}$, $M_{1,o}^t$, and $M_{2,o}^t$ remain mutually independent and uniform. Since $M_{1,e}^t$ is uniform and independent of $(K_{t-1},\mathcal H_{t-1})$, the ideal ciphertext $C_t=M_{1,e}^t\oplus K_{t-1}$ is uniform and independent of
\[
\left(M_{1,\mathrm{sec}}^t,K_t^{\rm full},M_2^t,M_{1,o}^t,M_{2,o}^t,\mathcal H_{t-1}\right).
\]
Consequently, conditional on $\mathcal H_{t-1}$, the five indices in \eqref{eq:selected-mapping} satisfy the joint independence and uniformity assumptions of Lemma~\ref{lem:selected-resolvability}. The lemma may therefore be applied with expectation only over $\mathsf C_t$; averaging the resulting conditional bound over $\mathcal H_{t-1}$ and the remaining codebooks gives the unconditional one-round bound.
\end{lemma}
\begin{IEEEproof}
For every value $h$ of $\mathcal H_{t-1}$ and every $c$ in the ciphertext group, we have
\begin{align}
&\Pr\{C_t=c\mid\mathcal H_{t-1}=h\}\notag\\
=&\sum_k\Pr\{K_{t-1}=k\mid h\}\Pr\{M_{1,e}^t=c\ominus k\}
=|G_n|^{-1}.
\end{align}
The same calculation after adjoining any collection of the fresh round-$t$ variables listed above gives the asserted joint factorization. Independence of $\mathsf C_t$ follows from the independent generation of the round codebooks.
\end{IEEEproof}

We apply Lemma~\ref{lem:selected-resolvability} conditionally as stated in Lemma~\ref{lem:current-round-index-factorization}, under the following identification of its abstract indices with the round-$t$ variables:
\begin{align}\label{eq:selected-mapping}
 A_1&=M_{1,\mathrm{sec}}^t,\quad\ B_1=(C_t,M_{1,o}^t);\notag\\
 A_2&=K_t^{\rm full}, \quad J_2=M_{2}^t,\quad B_2=M_{2,o}^t.
\end{align}
Under this identification, the effective averaging rates are $\widetilde R_{1,r}$ and $R_2+\widetilde R_{2,r}$. For every realization of $\mathcal H_{t-1}$, Lemma~\ref{lem:selected-resolvability} bounds the conditional expectation over the independent current codebook $\mathsf C_t$. The bound is independent of the realized history; the tower property therefore gives the same bound after averaging over the past and all other codebooks. Denote this one-round bound by $\delta_n$:
\begin{align}
\delta_n:=&e^{ns(E_s(P)-\widetilde R_{1,r}-\widetilde R_{2,r}-R_2)}\notag\\
&+e^{ns(C_s(P)-\widetilde R_{1,r})}
+e^{ns(D_s(P)-\widetilde R_{2,r}-R_2)}.
\label{eq:delta-block}
\end{align}

\begin{lemma}\Label{L3TA}
For every fixed $P\in\mathcal Q$, every fixed $s\in(0,1]$, and every integer $T\geq2$, there exists a reduced one-sided adaptive code $\Co_n^T$ of blocklength $nT$ and payload sizes $e^{n(T-1)R_1}$ and $e^{n(T-1)R_2}$ such that
\begin{align}
P_e^{nT}(\Co_n^T)
\leq 2T\Big[&e^{ns(\widetilde R_2+\widetilde R_{2,r}-B_s(P))}\notag\\
&+e^{ns(\widetilde R_1+\widetilde R_{1,r}-A_s(P))}
\Big],\label{eq:ReliabilityDA}\\
I(\mathbf M_1^{2:T};Z^{nT}\mid\Co_n^T)
&\leq 2T\delta_n+2\eta_{n,T}.\label{eq:SSecDA}
\end{align}
where $\mathbf M_1^{2:T}:=(M_1^2,\ldots,M_1^T)$, $M_1^t=(M_{1,\mathrm{sec}}^t, M_{1,e}^t)$,  $Z^{nT}=(Z^n[1],\ldots,Z^n[T])$, and $\eta_{n,T}$ is defined in \eqref{eq:eta-definition}. For fixed $T$, it decreases exponentially under the strict reliability conditions.
\end{lemma}
\begin{IEEEproof}
The proof is given in Appendix~\ref{APP: Proof of Lemma L3TA}.
\end{IEEEproof}

\subsection{Achievable secrecy region}
For each $P\in\mathcal Q_{\rm F}$, define the fixed-distribution
adaptive region
\begin{align}
\mathcal R_{\rm A}(P):=
\Bigl\{(R_1,R_2)\in\mathbb R_+^2:\;&
R_1\leq A(P),\notag\\
&R_2\leq B(P),\notag\\
&R_1\leq A(P)+B(P)-E(P),\notag\\
&R_1+R_2\leq A(P)+B(P)-C(P)
\Bigr\}.
\label{HH2AA}
\end{align}

\begin{lemma}\label{lem:A-achievable-region}
The following payload-rate region is achievable by the reduced
adaptive construction under the strong mixed-secrecy criterion:
\begin{equation}
\mathcal R_{\rm A}
:=
\overline{\operatorname{conv}}\!\left(
\bigcup_{P\in\mathcal Q_{\rm F}}
\mathcal R_{\rm A}(P)
\right).
\label{eq:A-achievable-region}
\end{equation}
\end{lemma}
\begin{IEEEproof}
\textit{Step 1: Exact projection for fixed $P$.}
Fix $P\in\mathcal Q_{\rm F}$. The auxiliary rates of the reduced
construction satisfy the following inequalities:
\begin{align}
\widetilde R_2+\widetilde R_{2,r}
&<B(P), \Label{G5AA}\\
\widetilde R_1+\widetilde R_{1,r}
&<A(P), \Label{G4AA}\\
\widetilde R_{1,r}
&>C(P), \Label{G1A}\\
\widetilde R_{2,r}+R_2
&>D(P), \Label{G2A}\\
\widetilde R_{1,r}+\widetilde R_{2,r}+R_2
&>E(P). \Label{G3A}
\end{align}
Together with \eqref{G6A}, \eqref{G8A}, and
\eqref{eq:reduced-rates}, these inequalities define the
fixed-$P$ auxiliary-rate system.

Introduce the following aggregate rates:
\begin{align}
x&:=R_{1,e}+R_{1,o},\notag\\
y&:=R_2+R_{2,o},\notag\\
e&:=R_{1,e}.
\label{eq:adaptive-aggregate-rates}
\end{align}
For fixed $(R_1,R_2,x,y,e)$, the key-size constraint requires
$R_{2,k}\geq e$. Increasing $R_{2,k}$ only tightens the reliability
constraint for User~2, so feasibility can be tested by setting
$R_{2,k}=e$. The reduced auxiliary-rate system is then equivalent to the following system of inequalities:
\begin{align}
x&>C(P),\notag\\
y&>D(P),\notag\\
x+y&>E(P),\notag\\
y&\geq R_2,\notag\\
R_1+x-e&<A(P),\notag\\
e+y&<B(P),\notag\\
0\leq e&\leq\min\{R_1,x\}.
\label{eq:adaptive-reduced}
\end{align}
For strictly feasible interior points, the interval condition for
$e$ is given by the following pair of inequalities:
\begin{align}
e&>\max\{0,R_1+x-A(P)\},\notag\\
e&<\min\{R_1,x,B(P)-y\}.
\label{eq:e-interval}
\end{align}
Eliminating $e$, $x$, and $y$ gives the following preliminary description of the projected region:
\begin{align}
R_1&<A(P),\notag\\
R_2&<B(P),\notag\\
R_1&<A(P)+B(P)\notag\\
&\quad-\max\{E(P),C(P)+D(P)\},\notag\\
R_1+R_2&<A(P)+B(P)-C(P).
\label{eq:A-fixed-P-projection-preliminary}
\end{align}
Since $V_1$ and $V_2$ are independent under every
$P\in\mathcal Q$, $E(P)$ admits the following decomposition:
\begin{align}
E(P)
={}&C(P)+D(P)\notag\\
&+I_P(V_1;V_2\mid Z)
\geq C(P)+D(P).
\label{eq:E-dominates-CD}
\end{align}
The third constraint in
\eqref{eq:A-fixed-P-projection-preliminary} therefore reduces to
\begin{equation}
R_1<A(P)+B(P)-E(P).
\end{equation}
Consequently, the closure of the fixed-$P$ projection is
$\mathcal R_{\rm A}(P)$.

\textit{Step 2: Finite-$s$ margins and fixed-round bounds.}
Consider a rate pair in the interior of
$\mathcal R_{\rm A}(P)$. Step~1 gives component rates for which
\eqref{G5AA}--\eqref{G3A} hold with a common positive slack. Since
$P$ is fixed, the continuity relations stated in the non-adaptive
exponential evaluation imply that a sufficiently small fixed
$s\in(0,1]$ can be chosen so that all of the following inequalities hold:
\begin{align}
\widetilde R_2+\widetilde R_{2,r}
&<B_s(P),\notag\\
\widetilde R_1+\widetilde R_{1,r}
&<A_s(P),\notag\\
\widetilde R_{1,r}
&>C_s(P),\notag\\
\widetilde R_{2,r}+R_2
&>D_s(P),\notag\\
\widetilde R_{1,r}+\widetilde R_{2,r}+R_2
&>E_s(P).
\label{eq:A-finite-s-system}
\end{align}
Lemma~\ref{L3TA} then gives exponentially decreasing error and
leakage for every fixed number of rounds $T$, including the
actual-to-ideal correction in \eqref{eq:eta-definition}.

\textit{Step 3: Growing rounds and padding for fixed $P$.}
To remove the initialization-rate loss and obtain codes for every
sufficiently large total blocklength $N$, choose the following round parameters:
\begin{align}
T_N&:=\lfloor N^{1/3}\rfloor,\notag\\
n_N&:=\lfloor N/T_N\rfloor,\notag\\
N'_N&:=n_NT_N.
\label{eq:growing-round-parameters}
\end{align}
Run the $T_N$-round construction for the first $N'_N$ channel
uses. During the remaining $r_N=N-N'_N<T_N$ uses, both users
transmit fixed input symbols and the decoders ignore the
corresponding outputs. By memorylessness, the padding output is
independent of the messages and the active-block output, so it does
not increase the leakage or the error probability. Furthermore, these parameters satisfy the following asymptotic relations:
\begin{align}
T_N&\longrightarrow\infty,\notag\\
n_N&\longrightarrow\infty,\notag\\
\frac{\ln T_N}{n_N}&\longrightarrow0,\notag\\
\frac{N'_N}{N}&\longrightarrow1.
\label{eq:growing-round-limits}
\end{align}
If $\alpha>0$ denotes the minimum of the two fixed finite-$s$
reliability margins multiplied by $s$, then
\eqref{eq:pbar-bound} gives
\begin{equation}
\bar p_{n_N,T_N}
\leq 2T_Ne^{-n_N\alpha}
\end{equation}
for all sufficiently large $N$. The message-alphabet size satisfies
\begin{align}
\log|\mathcal M_1^{T_N-1}|
={}&n_N(T_N-1)R_1\notag\\
&+o(n_NT_N),
\end{align}
Therefore, \eqref{eq:eta-definition} yields
\begin{align}
\eta_{n_N,T_N}
={}&O\!\left(n_NT_N^2e^{-n_N\alpha}\right)\notag\\
&+2h_2\!\left(2T_Ne^{-n_N\alpha}\right)
\longrightarrow0.
\end{align}
The terms $T_N\delta_{n_N}$ and the reliability bound also tend
to zero because $\ln T_N/n_N\to0$. Finally, the effective payload rate of User~$i$ converges to $R_i$, because
\begin{align}
\frac{n_N(T_N-1)R_i}{N}
={}&\frac{N'_N}{N}
\left(1-\frac1{T_N}\right)R_i\notag\\
&\longrightarrow R_i.
\end{align}
Thus every interior point of $\mathcal R_{\rm A}(P)$ is
achievable for the fixed distribution $P$. Boundary points follow
by choosing a sequence of achievable interior rate pairs converging
to the desired point.

\textit{Step 4: Union, time sharing, and closure.}
The code distribution may be chosen arbitrarily from
$\mathcal Q_{\rm F}$. Therefore, every rate pair in
\begin{equation}
\bigcup_{P\in\mathcal Q_{\rm F}}
\mathcal R_{\rm A}(P)
\end{equation}
is achievable by the reduced adaptive construction. Time sharing
among finitely many adaptive codes remains adaptive and gives the
convex hull of this union. A standard diagonal argument then gives
its closure. Thus every rate pair in $\mathcal R_{\rm A}$ is
achievable.
\end{IEEEproof}

\begin{remark}[Relation to the full construction]
Fix $P\in\mathcal Q_{\rm F}$. For code construction~I, define
$\widetilde R_i=R_{i,\mathrm{sec}}+R_{i,k}$ and
$\widetilde R_{i,r}=R_{i,e}+R_{i,o}$, with
$R_i=R_{i,\mathrm{sec}}+R_{i,e}$ and
$R_{i,e}\leq R_{i\oplus1,k}$. Its Shannon-information auxiliary-rate system consists of the following inequalities:
\begin{align}
\widetilde R_2+\widetilde R_{2,r}
&<B(P), \Label{G5AA-FA}\\
\widetilde R_1+\widetilde R_{1,r}
&<A(P), \Label{G4AA-FA}\\
\widetilde R_{1,r}
&>C(P), \Label{G1A-FA}\\
\widetilde R_{2,r}+R_{2,\mathrm{sec}}
&>D(P), \Label{G2A-FA}\\
\widetilde R_{1,r}+\widetilde R_{2,r}
+R_{2,\mathrm{sec}}
&>E(P). \Label{G3A-FA}
\end{align}
The same fixed-$P$ continuity argument, combined with
Lemma~\ref{L3TA-FA}, establishes the achievability of every
strictly feasible tuple in this system. Appendices~\ref{APP: FM-Full-Adaptive} and~\ref{APP: FM-OS-Adaptive} show that, for
each fixed $P\in\mathcal Q_{\rm F}$, the closures of the full and
reduced projections coincide. Hence their unions over
$P\in\mathcal Q_{\rm F}$, and therefore their closed union regions,
also coincide. The reduced construction thus has a simpler
one-sided structure without loss in the achievable rate region.
\end{remark}
\begin{corollary}\label{cor:comparison}
For every fixed input distribution $P\in\mathcal Q_{\rm F}$, both
fixed-distribution regions are nonempty and
$\mathcal R_{\rm N}(P)\subseteq\mathcal R_{\rm A}(P)$. The
individual bound on User~1's rate improves from
\begin{align}
\min\bigl\{A(P)-C(P),
A(P)+B(P)-E(P)\bigr\}
\end{align}
to
\begin{equation}
\min\bigl\{A(P),A(P)+B(P)-E(P)\bigr\}.
\end{equation}
Nevertheless, both regions have the following common maximum sum-rate:
\begin{align}
A(P)+B(P)
-\max\bigl\{C(P),E(P)-B(P)\bigr\}.
\label{eq:common-sum-rate}
\end{align}
Consequently, we have the following overall inclusion:
\begin{equation}
\mathcal R_{\rm N}\subseteq\mathcal R_{\rm A}.
\label{eq:overall-region-inclusion}
\end{equation}
\end{corollary}
\begin{IEEEproof}
The non-adaptive inequalities imply
\begin{align}
R_1&\leq A(P)-C(P)\leq A(P),\notag\\
R_1&\leq A(P)+B(P)-E(P),\notag\\
R_1+R_2&\leq A(P)+B(P)-C(P),
\end{align}
so $\mathcal R_{\rm N}(P)\subseteq\mathcal R_{\rm A}(P)$.
The non-adaptive region is a rectangle with upper-right corner
\begin{equation}
\bigl(A(P)-\max\{C(P),E(P)-B(P)\},B(P)\bigr),
\end{equation}
and hence its maximum sum-rate is
\eqref{eq:common-sum-rate}. For the adaptive region,
\begin{align*}
&\max_{(R_1,R_2)\in\mathcal R_{\rm A}(P)}(R_1+R_2)\\
={}&\min\Bigl\{
A(P)+B(P)-C(P),\notag\\
&\hspace{12mm}B(P)+\min\{A(P),A(P)+B(P)-E(P)\}
\Bigr\}\\
={}&A(P)+B(P)
-\max\{C(P),E(P)-B(P)\}.
\end{align*}
If $C(P)\geq E(P)-B(P)$, the point
$(A(P)-C(P),B(P))$ attains this value. If
$C(P)<E(P)-B(P)$, the point
$(A(P)+B(P)-E(P),B(P))$ attains it. Both points satisfy all
adaptive inequalities in their respective cases.
Since $\mathcal R_{\rm N}(P)\subseteq
\mathcal R_{\rm A}(P)$ for every $P\in\mathcal Q_{\rm F}$, this
inclusion is preserved under unions, convex hulls, and closures.
Hence $\mathcal R_{\rm N}\subseteq\mathcal R_{\rm A}$.
\end{IEEEproof}
\section{Comparisons among achievable regions}\label{sec:comparisons}
The comparisons below concern construction-specific inner bounds, not capacity regions. Every fixed-distribution statement uses the strict feasibility conditions proved above. In particular, closure of a fixed-distribution projection does not add points when its strict auxiliary-rate system is empty.

\subsection{Weak- and strong-secrecy inner bounds}
Set $V_i=X_i$ and fix a product distribution of the form
\begin{equation}
P=P_{X_1}P_{X_2}.
\end{equation}
Define the product-input information quantities as follows:
\begin{align}
A_X(P)&:=I_P(Y_2;X_1\mid X_2),\notag\\
B_X(P)&:=I_P(Y_1;X_2\mid X_1),\notag\\
C_X(P)&:=I_P(Z;X_1),\notag\\
D_X(P)&:=I_P(Z;X_2),\notag\\
E_X(P)&:=I_P(Z;X_1,X_2).
\label{eq:product-information-quantities}
\end{align}
These are the specializations of the general information quantities
to $V_i=X_i$. In particular, the following identities hold:
\begin{align}
A_X(P)&=A(P),&B_X(P)&=B(P),\notag\\
C_X(P)&=C(P),&D_X(P)&=D(P),\notag\\
E_X(P)&=E(P).
\label{eq:product-general-identification}
\end{align}
The weak one-sided inner bound of \cite[Theorem~1]{QCHT2016} can be
written as
\begin{align}
\mathcal R_{\rm W}(P)
:=\Bigl\{(R_1,R_2)\in\mathbb R_+^2:&
R_2\leq B_X(P),\label{eq:weak-region-product}\\
&R_1\leq A_X(P)-C_X(P)\notag\\
&R_1\leq A_X(P)+R_2-E_X(P)\Bigr\}.\notag
\end{align}
Here product inputs give
$I_P(X_2;Z\mid X_1)=E_X(P)-C_X(P)$. Under the following strict feasibility conditions:
\begin{align}
C_X(P)&<A_X(P),\notag\\
D_X(P)&<B_X(P),\notag\\
E_X(P)&<A_X(P)+B_X(P),
\label{eq:product-strict-feasibility}
\end{align}
the present non-adaptive strong-secrecy inner bound is the following rate region:
\begin{align}
\mathcal R_{\rm N}(P)
=\Bigl\{(R_1,R_2)\in\mathbb R_+^2:\;&
R_2\leq B_X(P),\label{eq:strong-product-region}\\
&R_1\leq A_X(P)-C_X(P),\notag\\
&R_1\leq A_X(P)+B_X(P)-E_X(P)
\Bigr\}.\notag
\end{align}
If a displayed upper bound is negative, intersection with
$\mathbb R_+^2$ makes the region empty.

\begin{proposition}\label{prop:weak-strong-containment}
For every product distribution $P=P_{X_1}P_{X_2}$ satisfying
\eqref{eq:product-strict-feasibility},
\begin{equation}
\mathcal R_{\rm W}(P)\subseteq\mathcal R_{\rm N}(P).
\end{equation}
\end{proposition}
\begin{IEEEproof}
The weak bound directly gives $R_2\leq B_X(P)$ and
$R_1\leq A_X(P)-C_X(P).$ 
Following the fact that $A_X(P)+R_2-E_X(P)\leq A_X(P)+B_X(P)-E_X(P)$ as $R_2\leq B_X(P),$ every weak-bound point satisfies \eqref{eq:strong-product-region}.
Thus every weak-bound point satisfies
\eqref{eq:strong-product-region}.
\end{IEEEproof}
The same containment persists after taking unions over strictly
feasible product distributions, convex hulls, and closures on both
sides. This does not assert fixed-distribution achievability for a
distribution at which any condition in
\eqref{eq:product-strict-feasibility} fails; a boundary point added by
closure is justified only as a limit of rate points obtained from
strictly feasible distributions.

\subsection{Non-adaptive and key-exchange inner bounds}
For a product distribution $P=P_{X_1}P_{X_2}$ satisfying
\eqref{eq:product-strict-feasibility}, the key-exchange region is
\begin{align}
\mathcal R_{\rm A}(P)
=&\Bigl\{(R_1,R_2)\in\mathbb R_+^2:
R_1\leq A_X(P),\notag\\
&R_2\leq B_X(P),\notag\\
&R_1\leq A_X(P)+B_X(P)-E_X(P),\notag\\
&R_1+R_2\leq A_X(P)+B_X(P)-C_X(P)
\Bigr\}.
\label{eq:adaptive-product-region}
\end{align}
Corollary~\ref{cor:comparison}, specialized to $V_i=X_i$, gives
\begin{equation}
\mathcal R_{\rm N}(P)\subseteq\mathcal R_{\rm A}(P).
\end{equation}
The inclusion may be strict in the individual-rate direction, while
the maximum sum-rates agree by \eqref{eq:common-sum-rate} and
\eqref{eq:product-general-identification}. 

\section{Conclusion}\label{Sec: Conclusion}
We derived non-adaptive and key-exchange-based adaptive achievable regions for the TW-WC under strong one-sided secrecy. For every fixed distribution in $\mathcal Q_{\rm F}$, key exchange can improve User~1's individual rate while leaving the maximum sum-rate unchanged. The overall adaptive region 
includes the non-adaptive achievable region. 
For every fixed number of rounds, error and leakage decrease exponentially in the per-round blocklength. The reduced-construction proof uses conditional selected-subindex resolvability, a round-wise factorization, a one-time-pad inequality with side information, and an ideal-to-actual transfer. A growing-round sequence with fixed-input padding removes the initialization-rate loss and yields vanishing error and leakage for every sufficiently large total blocklength. No positive exponent per total blocklength is claimed for this growing-round sequence. The full and reduced auxiliary-rate systems have the same projection, while the main achievability proof relies on the reduced construction.
\section*{Acknowledgements} 
During the preparation of this manuscript, the authors used Microsoft Copilot to assist with language editing, the organization of the text, and the presentation and verification of certain mathematical derivations.
All AI-assisted material was critically reviewed, verified, and revised by the authors, who take full responsibility for the accuracy and integrity of the manuscript.

\appendices
\section{Additional details for Lemma \ref{lem:NA-achievable-region}}\Label{APP: Proof of Lemma NA region}
The main text establishes achievability by connecting the Shannon-information auxiliary-rate system to the finite-sss exponential bounds. This appendix supplies the corresponding Fourier--Motzkin elimination.
For completeness, Appendix~\ref{APP: FM-Non-Adaptive} provides the Fourier--Motzkin elimination leading to the stated Shannon-information projection.

\section{Proof of Lemma \ref{lem:selected-resolvability}}\Label{APP: Proof of Lemma Selected-Resolvability}

For fixed $(a_1,a_2)$, the conditional output distribution is the uniform mixture over $(b_1,j_2,b_2)$. Applying \cite[Lemma~2]{HC2023}, as specialized in the derivation of \cite[Eqs.~(58)--(59)]{HC2023}, with effective randomization sizes $e^{nR_{B_1}}$ and $e^{n(R_{J_2}+R_{B_2})}$ yields the three nonempty-subset terms in \eqref{eq:selected-resolvability}. Averaging over $(a_1,a_2)$ and using the divergence decomposition
\begin{align*}
	&\sum_{a_1,a_2}P(a_1,a_2)
	D(P_{Z^n\mid a_1,a_2,\mathsf C}\Vert P_{Z^n})\\
	&\quad=I(A_1,A_2;Z^n\mid\mathsf C)
	+D(P_{Z^n\mid\mathsf C}\Vert P_{Z^n})
\end{align*}
completes the proof after dropping the last nonnegative term.

\section{Proof of Lemma \ref{lem:otp-side-information}}\Label{APP: Proof of Lemma Otp-side-information}
Fix $w$ in the support of $W$. The stated independence assumptions
imply
\begin{align*}
I(U;C,S\mid W=w)
&=I(U;C\mid S,W=w)\\
&=H(C\mid S,W=w)-H(K\mid S,W=w)\\
&\leq\log|G|-H(K\mid S,W=w)\\
&=I(K;S\mid W=w).
\end{align*}
Averaging both sides over $W$ proves
\eqref{eq:otp-side-information}. The projection statement follows
from uniformity under a surjective homomorphism and data processing.
\section{Proof of Lemma \ref{lem:ideal-real-transfer}}\Label{APP: Proof of Lemma Ideal-real-transfer}
For each fixed $\boldsymbol c$, the coupled processes are identical unless a key used by a later encoder has been decoded incorrectly, which proves the total-variation assertion by the coupling inequality. Define the good and bad codebook sets by
\[
\mathcal G:=\{\boldsymbol c:p_{n,T}(\boldsymbol c)\leq1/2\},
\qquad \mathcal G^{\rm c}:=\{\boldsymbol c:p_{n,T}(\boldsymbol c)>1/2\}.
\]
On $\mathcal G$, the payload marginal is the same uniform distribution in the two processes. We use the following finite-alphabet conditional-entropy continuity inequality:
\begin{align*}
|H_P(X\mid Y)-H_Q(X\mid Y)|&\leq 2\varepsilon\log|\mathcal X|+2h_2(\varepsilon),\\
\|P_{XY}-Q_{XY}\|_{\rm TV}&\leq\varepsilon\leq\tfrac12.
\end{align*}
Hence the common term $H(\mathbf M_1^{2:T})$ cancels when the two mutual informations are compared, and the inequality gives
\begin{align*}
	I_{\rm real}(\mathbf M_1^{2:T};Z^{nT}\mid\boldsymbol c)
	&\leq I_{\rm ideal}(\mathbf M_1^{2:T};Z^{nT}\mid\boldsymbol c)\\
	&\quad+2p_{n,T}(\boldsymbol c)\log|\mathcal{M}_1^{T-1}|
	+2h_2(p_{n,T}(\boldsymbol c)).
\end{align*}
On $\mathcal G^{\rm c}$, the trivial bounds
$0\leq I_{\rm real},I_{\rm ideal}\leq\log|\mathcal{M}_1^{T-1}|$
give
$I_{\rm real}-I_{\rm ideal}\leq\log|\mathcal{M}_1^{T-1}|$.
Markov's inequality yields
$\Pr\{\boldsymbol{\mathsf C}\in\mathcal G^{\rm c}\}\leq2\bar p_{n,T}$.
After averaging the good-set bound and the bad-set trivial bound, the two logarithmic contributions are together at most
$4\bar p_{n,T}\log|\mathcal{M}_1^{T-1}|$.
For the entropy contribution, set
$Q(\boldsymbol c):=p_{n,T}(\boldsymbol c)\mathbf{1}_{\{\boldsymbol c\in\mathcal G\}}$.
Concavity of $h_2$, $\mathbb E Q\leq\bar p_{n,T}\leq1/2$, and monotonicity of $h_2$ on $[0,1/2]$ give
$\mathbb E h_2(Q)\leq h_2(\mathbb E Q)\leq h_2(\bar p_{n,T})$.
This proves \eqref{eq:ideal-real-transfer}--\eqref{eq:eta-definition} without requiring
$p_{n,T}(\boldsymbol c)\leq1/2$ for every codebook realization.
Finally, $F_{n,T}$ is contained in the union of the corresponding nominal-index decoding-error events, so a union bound over the rounds and the one-block decoding estimate give \eqref{eq:pbar-bound}.

\section{Proof of Lemma \ref{L3TA}}\Label{APP: Proof of Lemma L3TA}
Generate the $T$ round codebooks independently. The one-block decoding estimate and a union bound give the ensemble-average reliability bound without the leading factor~2.

We first analyze the coupled ideal-key process, in which the true previous-round key is used in the ciphertext. Let $\boldsymbol{\mathsf C}$ denote the full codebook collection. We track the accumulated leakage, including the key required in the next round, through the quantity
\begin{equation}
	L_t:=\mathbb E_{\boldsymbol{\mathsf C}}
	I( M_{1}[2:t],K_t^{\rm full};Z^n[1:t]\mid\boldsymbol{\mathsf C}),
\end{equation}
where $M_{1}[2:t]=(M_1^2,\ldots, M_1^t),$ $M_1^i=(M_{1,\mathrm{sec}}^i, M_{1,e}^i)$ for $i\in [2:t],$ and $Z^n[1:t]=(Z^n[1],\ldots,Z^n[t]).$
In round~1, Lemma~\ref{lem:selected-resolvability} protects $(D_{1,\mathrm{sec}},K_1^{\rm full})$. Discarding the dummy coordinate gives $L_1\leq\delta_n$.

For $t\geq2$, introduce the following notation:
\begin{align*}
 W_t&:=M_1[2:t-1],& K^-_t&:=K_{t-1},& U_t&:=M_{1,e}^t,\\
 S_t&:=(M_{1,\mathrm{sec}}^t,K_t^{\rm full}),& E^{t-1}&:=Z^n[1:t-1],& Z_t&:=Z^n[t].
\end{align*}
Here $W_t$ denotes the accumulated payload and is distinct from the fixed input distribution $P$. 
Write $\boldsymbol{\mathsf C}=(\boldsymbol{\mathsf C}_{<t},\mathsf C_t,\boldsymbol{\mathsf C}_{>t})$.

\begin{lemma}[Round-wise conditional structure]\label{lem:roundwise-conditional-structure}
For the ideal process and every fixed realization of $\boldsymbol{\mathsf C}$, the following properties hold:
\begin{align}
 S_t&\perp (W_t,K^-_t,E^{t-1})\mid\boldsymbol{\mathsf C},\label{eq:fresh-pair-independence}\\
 C_t&\perp(W_t,K^-_t,E^{t-1},S_t)\mid\boldsymbol{\mathsf C},\label{eq:ciphertext-independence}\\
 P_{Z_t\mid W_t,E^{t-1},S_t,C_t,\boldsymbol{\mathsf C}} &=P_{Z_t\mid S_t,C_t,\boldsymbol{\mathsf C}},\label{eq:round-kernel}\\
 P_{Z_t\mid W_t,E^{t-1},S_t,\boldsymbol{\mathsf C}} &=P_{Z_t\mid S_t,\boldsymbol{\mathsf C}},\label{eq:round-averaged-kernel}
\end{align}
In addition, the following Markov chain holds:
\begin{equation}\label{eq:round-otp-markov}
 U_t-(C_t,E^{t-1},W_t,S_t,\boldsymbol{\mathsf C})-Z_t.
\end{equation}
Consequently, the following two identities hold:
\begin{align}
 I(S_t;E^{t-1},Z_t\mid\boldsymbol{\mathsf C})
 &=I(S_t;Z_t\mid\boldsymbol{\mathsf C}),\label{eq:round-fresh-identity}\\
 I(W_t;E^{t-1},Z_t\mid S_t,\boldsymbol{\mathsf C})
 &=I(W_t;E^{t-1}\mid\boldsymbol{\mathsf C}).\label{eq:round-old-identity}
\end{align}
\end{lemma}
\begin{IEEEproof}
The fresh pair $S_t$ and all fresh open variables are generated independently of preceding-round variables and independently of the codebooks, which proves \eqref{eq:fresh-pair-independence}. By Lemma~\ref{lem:current-round-index-factorization}, the fresh uniform variable $U_t$ makes $C_t=U_t\oplus K^-_t$ uniform and independent of $(W_t,K^-_t,E^{t-1},S_t)$, proving \eqref{eq:ciphertext-independence}. After averaging the fresh open indices, channel memorylessness and independence of $\mathsf C_t$ give \eqref{eq:round-kernel}. Averaging this kernel over the uniform $C_t$ gives \eqref{eq:round-averaged-kernel}. Once $(C_t,S_t,\mathsf C_t)$ and the fresh open indices have selected the current channel inputs, $Z_t$ has no further dependence on $U_t$, which gives \eqref{eq:round-otp-markov}.

For \eqref{eq:round-fresh-identity}, use the chain rule, $S_t\perp E^{t-1}\mid\boldsymbol{\mathsf C}$, and \eqref{eq:round-averaged-kernel}. For \eqref{eq:round-old-identity}, first remove $S_t$ using \eqref{eq:fresh-pair-independence}; then remove $Z_t$ because \eqref{eq:round-averaged-kernel} implies $Z_t\perp(W_t,E^{t-1})\mid(S_t,\boldsymbol{\mathsf C})$.
\end{IEEEproof}

The chain rule separates the fresh protected pair from the accumulated payload and current encrypted component as follows:
\begin{align}
	I(W_t,U_t,S_t;E^{t-1},Z_t\mid\boldsymbol{\mathsf C})
	&=I(S_t;E^{t-1},Z_t\mid\boldsymbol{\mathsf C})\notag\\
	&\quad+I(W_t,U_t;E^{t-1},Z_t\mid S_t,\boldsymbol{\mathsf C}).\label{eq:recursion-split}
\end{align}
Lemma~\ref{lem:roundwise-conditional-structure} gives the identity
\begin{equation}\label{eq:fresh-term}
	I(S_t;E^{t-1},Z_t\mid\boldsymbol{\mathsf C})
	=I(S_t;Z_t\mid\boldsymbol{\mathsf C}).
\end{equation}
Condition on the past and on all codebooks except $\mathsf C_t$. Since $S_t$ is independent of the past history conditional on the codebooks, adjoining that history can only increase the mutual information relevant to the bound. More precisely,
\begin{align}
&I(S_t;Z_t\mid\boldsymbol{\mathsf C})\notag\\
&\quad\leq I(S_t;Z_t,\mathcal H_{t-1}\mid\boldsymbol{\mathsf C})\notag\\
&\quad=I(S_t;Z_t\mid\mathcal H_{t-1},\boldsymbol{\mathsf C}),
\label{eq:fresh-term-history}
\end{align}
where the equality uses $S_t\perp\mathcal H_{t-1}\mid\boldsymbol{\mathsf C}$. Lemma~\ref{lem:current-round-index-factorization} then permits Lemma~\ref{lem:selected-resolvability} to be applied to the conditional expectation over $\mathsf C_t$. The tower property therefore bounds the expectation of \eqref{eq:fresh-term}, via \eqref{eq:fresh-term-history}, by $\delta_n$.

The chain rule decomposes the second term on the right-hand side of \eqref{eq:recursion-split} as follows:
\begin{align}
	I(W_t,U_t;E^{t-1},Z_t\mid S_t,\boldsymbol{\mathsf C})
	&=I(W_t;E^{t-1},Z_t\mid S_t,\boldsymbol{\mathsf C})\notag\\
	&\quad+I(U_t;E^{t-1},Z_t\mid W_t,S_t,\boldsymbol{\mathsf C}).\label{eq:old-otp-split}
\end{align}
The old-term identity \eqref{eq:round-old-identity} gives
\begin{equation}
	I(W_t;E^{t-1},Z_t\mid S_t,\boldsymbol{\mathsf C})
	=I(W_t;E^{t-1}\mid\boldsymbol{\mathsf C}).\label{eq:old-term}
\end{equation}
The Markov relation \eqref{eq:round-otp-markov}, data processing, Lemma~\ref{lem:otp-side-information} conditioned on $(W_t,S_t,\boldsymbol{\mathsf C})$, and \eqref{eq:fresh-pair-independence} give the OTP-absorption bound
\begin{align}
	I(U_t;E^{t-1},Z_t\mid W_t,S_t,\boldsymbol{\mathsf C})
	&\leq I(U_t;C_t,E^{t-1}\mid W_t,S_t,\boldsymbol{\mathsf C})\notag\\
	&\leq I(K^-_t;E^{t-1}\mid W_t,\boldsymbol{\mathsf C}).\label{eq:otp-absorption}
\end{align}
Since $K^-_t$ is independent of $W_t$ before $E^{t-1}$ is observed, the chain rule gives the identity
\begin{align}
	I(W_t;E^{t-1}\mid\boldsymbol{\mathsf C})
	+I(K^-_t;E^{t-1}\mid W_t,\boldsymbol{\mathsf C})
	=I(W_t,K^-_t;E^{t-1}\mid\boldsymbol{\mathsf C}).\label{eq:absorption-identity}
\end{align}
Combining \eqref{eq:recursion-split}--\eqref{eq:absorption-identity} and then taking iterated expectation over the current and past codebooks gives the recursion
\begin{equation}\label{eq:ensemble-recursion}
	L_t\leq L_{t-1}+\delta_n.
\end{equation}
Thus $L_T\leq T\delta_n$, and data processing after discarding $K_T^{\rm full}$ gives the ensemble-average ideal-process payload leakage bound $T\delta_n$. Lemma~\ref{lem:ideal-real-transfer} therefore gives the ensemble-average actual-process bound $T\delta_n+\eta_{n,T}$.

Let $a_n$ denote the displayed ensemble-average reliability bound without the leading factor~2, and let $b_n:=T\delta_n+\eta_{n,T}$ denote the ensemble-average actual-process leakage bound. If both are positive, define the normalized sum criterion
\[
X(\boldsymbol{\mathsf C})=
\frac{P_e(\boldsymbol{\mathsf C})}{a_n}
+\frac{I(\mathbf M_1^{2:T};Z^{nT}\mid\boldsymbol{\mathsf C})}{b_n}.
\]
Since $\mathbb E X\leq2$, there exists a deterministic codebook realization for which $X\leq2$. For this realization, each criterion is at most twice its corresponding ensemble average. In particular, the leakage is at most $2T\delta_n+2\eta_{n,T}$. A zero denominator means that the associated nonnegative criterion vanishes almost surely and is handled directly. This proves \eqref{eq:ReliabilityDA} and \eqref{eq:SSecDA}.

\section{Additional details for Lemma \ref{lem:A-achievable-region}}\Label{APP: Proof of Lemma A-Projection}
The main proof of Lemma~\ref{lem:A-achievable-region} gives the reduced-system projection and the continuity-based achievability argument. 
Appendices~\ref{APP: FM-Full-Adaptive} and~\ref{APP: FM-OS-Adaptive} provide the corresponding Fourier--Motzkin eliminations and show that the full and reduced constructions have the same projected region.

\section{Fourier-Motzkin elimination: non-adaptive region}\Label{APP: FM-Non-Adaptive}
Fix $P\in\mathcal Q$. All information quantities in this appendix
are evaluated under this fixed distribution and retain their explicit
dependence on $P$. To derive the secrecy region by non-adaptive coding, recall that we have the following rate constraints:
\allowdisplaybreaks
\begin{align}
	R_{1} + R_{1,r} <& A(P),\Label{FM_OS_1}\\
	R_{2} + R_{2,r} <& B(P), \Label{FM_OS_2}\\
	R_{1,r} > & C(P), \Label{FM_OS_3}\\
	R_{2} + R_{2,r} > & D(P),\Label{FM_OS_4}\\
	R_{1,r} + R_{2} +R_{2,r} > & E(P). \Label{FM_OS_5}
\end{align}
Eliminating $R_{1,r}$ from \eqref{FM_OS_1}, \eqref{FM_OS_3}, and \eqref{FM_OS_5} yields the following constraints:
\begin{align}
	R_{1}  <& A(P)-C(P),\Label{FM_OS_6}\\
	R_{1}-R_{2} - R_{2,r}  <& A(P)-E(P).\Label{FM_OS_7}
\end{align}
Eliminating $R_{2,r}$ from \eqref{FM_OS_2}, \eqref{FM_OS_4}, and \eqref{FM_OS_7} yields the following constraints:
\begin{align}
	 D(P) <& B(P), \Label{FM_OS_8}\\
	R_{2} <& B(P), \Label{FM_OS_9}\\
	R_{1} <& A(P)+B(P)-E(P).\Label{FM_OS_10}
\end{align} 
Therefore, under the exact strict feasibility conditions $C(P)<A(P)$, $D(P)<B(P)$, and $E(P)<A(P)+B(P)$, the following region is achievable 
\allowdisplaybreaks
\begin{align*}
	R_{1}  <& A(P)-\max\{C(P), E(P)-B(P)\},\\
	R_{2} <& B(P).
\end{align*}

\section{Fourier-Motzkin elimination: adaptive region}\Label{APP: FM-Full-Adaptive}
Fix $P\in\mathcal Q$. All information quantities in this appendix
are evaluated under this fixed distribution and retain their explicit
dependence on $P$. To derive the one-sided secrecy region by the adaptive key-exchange construction, recall that we have the following rate constraints. All component rates appearing below are nonnegative.
\allowdisplaybreaks
\begin{align}
	R_1 =& R_{1,\mathrm{sec}} + R_{1,e}, \Label{FM_Ind_8}\\ 
	R_2 =& R_{2,\mathrm{sec}} + R_{2,e},\Label{FM_Ind_9}\\
	R_{2,\mathrm{sec}} + R_{2,k}+R_{2,o} + R_{2,e} <& B(P), \Label{FM_Ind_7}\\
	R_{1,\mathrm{sec}} + R_{1,k}+R_{1,o} + R_{1,e}  <& A(P),\Label{FM_Ind_6}\\
	R_{1,e} \leq & R_{2,k}, \Label{FM_Ind_4}\\
	R_{2,e} \leq & R_{1,k}, \Label{FM_Ind_5}\\
	R_{1,o} + R_{1,e} > & C(P), \Label{FM_Ind_1}\\
	R_{2,o} + R_{2,e} + R_{2,\mathrm{sec}} > & D(P),\Label{FM_FA_2}\\
	R_{1,o} + R_{1,e} +R_{2,o} + R_{2,e}+ R_{2,\mathrm{sec}}> & E(P). \Label{FM_FA_3}
\end{align}
First consider \eqref{FM_Ind_6}, \eqref{FM_Ind_1} and \eqref{FM_FA_3} to remove $R_{1,o}.$ We obtain the following constraints:
\begin{align}
	R_{1,\mathrm{sec}} + R_{1,k}+ R_{1,e}  <& A(P),\Label{FM_Ind_11}\\
	R_{1,\mathrm{sec}} + R_{1,k}  <& A(P)-C(P),\Label{FM_Ind_12}\\
	R_{2,o} + R_{2,e}+ R_{2,\mathrm{sec}} -	R_{1,\mathrm{sec}} - R_{1,k}> & E(P)-A(P). \Label{FM_Ind_13}
\end{align}
Consider \eqref{FM_Ind_7} , \eqref{FM_FA_2} and \eqref{FM_Ind_13} to remove $R_{2,o}.$ We obtain the following constraints:
\begin{align}
	R_{2,\mathrm{sec}} + R_{2,k}+R_{2,e} <& B(P), \Label{FM_Ind_14}\\
	R_{2,k}  <& B(P)-D(P), \Label{FM_Ind_15}\\
	R_{1,k} + R_{2,k} + R_{1,\mathrm{sec}}  <& A(P)+B(P)-E(P).\Label{FM_Ind_16}
\end{align} 
Consider \eqref{FM_Ind_5}, \eqref{FM_Ind_11}, \eqref{FM_Ind_12} and \eqref{FM_Ind_16} to remove $R_{1,k}.$ We obtain the following constraints:
\begin{align}
	R_{1,\mathrm{sec}} + R_{1,e}+ R_{2,e}  <& A(P),\Label{FM_Ind_17}\\
	R_{1,\mathrm{sec}} + R_{2,e}  <& A(P)-C(P),\Label{FM_Ind_18}\\
	R_{2,e} + R_{2,k} + R_{1,\mathrm{sec}}  <& A(P)+B(P)-E(P).\Label{FM_Ind_19}
\end{align} 
Consider \eqref{FM_Ind_4}, \eqref{FM_Ind_14} and \eqref{FM_Ind_15}  and \eqref{FM_Ind_19} to remove $R_{2,k}.$ We obtain the following constraints:
\begin{align}
	R_{2,\mathrm{sec}} + R_{2,e}+R_{1,e} <& B(P), \Label{FM_Ind_20}\\
	R_{1,e}  <&  B(P)-D(P), \Label{FM_Ind_21}\\
	R_{1,e} + R_{2,e} + R_{1,\mathrm{sec}} <& A(P)+B(P)-E(P).\Label{FM_Ind_22}
\end{align} 
Next, we remove $R_{1,e}$ and $R_{2,e}$ replacing them by $R_1-R_{1,\mathrm{sec}}$ and $R_2-R_{2,\mathrm{sec}}$ (according to \eqref{FM_Ind_8} and \eqref{FM_Ind_9}), respectively, in \eqref{FM_Ind_17}, \eqref{FM_Ind_18}, \eqref{FM_Ind_20}, \eqref{FM_Ind_21} and \eqref{FM_Ind_22}. Together with the non-negativity of $R_{1,e}$ and $R_{2,e}$, we obtain
\allowdisplaybreaks
\begin{align}
	R_1 \geq & R_{1,\mathrm{sec}}, \Label{FM_Ind_8C}\\ 
	R_2 \geq & R_{2,\mathrm{sec}},\Label{FM_Ind_9C}\\
	R_1+ R_2-R_{2,\mathrm{sec}}  <& A(P),\Label{FM_Ind_17C}\\
	R_{1,\mathrm{sec}} + R_2-R_{2,\mathrm{sec}}  <&A(P)-C(P),\Label{FM_Ind_18C}\\
	R_2 + R_1-R_{1,\mathrm{sec}} <& B(P), \Label{FM_Ind_20C}\\
	R_1-R_{1,\mathrm{sec}}  <& B(P)-D(P), \Label{FM_Ind_21C}\\
	R_1+ R_2- R_{2,\mathrm{sec}}  <& A(P)+B(P)-E(P).\Label{FM_Ind_22C}
\end{align}
Consider  \eqref{FM_Ind_8C}, \eqref{FM_Ind_18C},  \eqref{FM_Ind_20C} and \eqref{FM_Ind_21C} to remove $R_{1,\mathrm{sec}}.$ We obtain the following constraints:
\allowdisplaybreaks
\begin{align}
	R_2<& B(P), \Label{FM_Ind_23}\\
	D(P) <& B(P), \Label{FM_Ind_24}\\
	R_2-R_{2,\mathrm{sec}}  <&  A(P)-C(P),\Label{FM_Ind_26}\\
	R_1+2R_2-R_{2,\mathrm{sec}}  <& A(P)+B(P)-C(P),\Label{FM_Ind_27}\\
	R_1+R_2-R_{2,\mathrm{sec}}<& A(P)+B(P)-C(P)-D(P),\Label{FM_Ind_28}
\end{align}
Consider \eqref{FM_Ind_9C}, \eqref{FM_Ind_17C}, \eqref{FM_Ind_22C}, \eqref{FM_Ind_26}, \eqref{FM_Ind_27} and \eqref{FM_Ind_28} to remove $R_{2,\mathrm{sec}}.$ We obtain the following constraints:
\allowdisplaybreaks
\begin{align}
	R_1 <& A(P),\Label{FM_Ind_30}\\
	R_1 <& A(P)+B(P)-E(P),\Label{FM_Ind_37}\\
	C(P)< & A(P),\Label{FM_Ind_31}\\
	R_1 +R_2< &  A(P)+B(P)-C(P),\Label{FM_Ind_33}\\
	R_1 < & A(P)+B(P)-C(P)-D(P). \Label{FM_Ind_34}
\end{align}
Therefore, under the exact strict feasibility conditions $D(P) < B(P)$, $C(P) < A(P)$, and $E(P)<A(P)+B(P),$ the projection of this auxiliary-rate system is 
\allowdisplaybreaks
\begin{align*}
	R_{1}  <& A(P),\\
	R_{1}  <& A(P)+B(P)-\max\{E(P), C(P)+D(P)\},\\
	R_{2} <& B(P),\\
	R_1 + R_2 <& A(P)+B(P)-C(P).
\end{align*}

\section{Fourier-Motzkin elimination: adaptive region with one-sided reduction}\Label{APP: FM-OS-Adaptive}
Fix $P\in\mathcal Q$. All information quantities in this appendix
are evaluated under this fixed distribution and retain their explicit
dependence on $P$. To derive the one-sided secrecy region by the adaptive key-exchange construction with one-sided reduction, recall that we have the following rate constraints:
\allowdisplaybreaks
\begin{align}
	R_1 =& R_{1,\mathrm{sec}} + R_{1,e}, \Label{FM_OSA_8}\\ 
	R_2 =& R_{2,\mathrm{sec}} \Label{FM_OSA_9}\\
	R_{2,\mathrm{sec}} + R_{2,k}+R_{2,o} <& B(P), \Label{FM_OSA_7}\\
	R_{1,\mathrm{sec}} + R_{1,o} + R_{1,e}  <& A(P),\Label{FM_OSA_6}\\
	R_{1,e} \leq & R_{2,k}, \Label{FM_OSA_4}\\
	R_{1,o} + R_{1,e} > & C(P), \Label{FM_OSA_1}\\
	R_{2,o} + R_{2,\mathrm{sec}} > & D(P),\Label{FM_OSA_2}\\
	R_{1,o} + R_{1,e} +R_{2,o} + R_{2,\mathrm{sec}}> & E(P). \Label{FM_OSA_3}
\end{align}
First consider \eqref{FM_OSA_6}, \eqref{FM_OSA_1} and \eqref{FM_OSA_3} to remove $R_{1,o}.$ We obtain the following constraints:
\begin{align}
	R_{1,\mathrm{sec}} + R_{1,e}  <& A(P),\Label{FM_OSA_11}\\
	R_{1,\mathrm{sec}}  			<& A(P)-C(P),\Label{FM_OSA_12}\\
	R_{2,o} + R_{2,\mathrm{sec}} -	R_{1,\mathrm{sec}} > & E(P)-A(P). \Label{FM_OSA_13}
\end{align}
Consider \eqref{FM_OSA_7} , \eqref{FM_OSA_2} and \eqref{FM_OSA_13} to remove $R_{2,o}.$ We obtain the following constraints:
\begin{align}
	R_{2,\mathrm{sec}} + R_{2,k}<& B(P), \Label{FM_OSA_14}\\
	R_{2,k}  <& B(P)-D(P), \Label{FM_OSA_15}\\
	R_{2,k} + R_{1,\mathrm{sec}}  <& A(P)+B(P)-E(P).\Label{FM_OSA_16}
\end{align} 

Consider \eqref{FM_OSA_4}, \eqref{FM_OSA_14} and \eqref{FM_OSA_15}  and \eqref{FM_OSA_16} to remove $R_{2,k}.$ We obtain the following constraints:
\begin{align}
	R_{2,\mathrm{sec}} + R_{1,e} <& B(P), \Label{FM_OSA_20}\\
	R_{1,e}  <&  B(P)-D(P), \Label{FM_OSA_21}\\
	R_{1,e} + R_{1,\mathrm{sec}}  <& A(P)+B(P)-E(P).\Label{FM_OSA_22}
\end{align} 

Next, using $R_{1,e}=R_1-R_{1,\mathrm{sec}}$ from \eqref{FM_OSA_8} and $R_{2,\mathrm{sec}}=R_2$ from \eqref{FM_OSA_9}, we eliminate $R_{1,e}$ and $R_{2,\mathrm{sec}}$ in \eqref{FM_OSA_11}, \eqref{FM_OSA_12}, \eqref{FM_OSA_20}, \eqref{FM_OSA_21}, and \eqref{FM_OSA_22}. Together with $R_{1,e}\geq 0$, equivalently $R_{1,\mathrm{sec}}\leq R_1$, we obtain
\allowdisplaybreaks
\begin{align}
	R_1  <& A(P),\Label{FM_OSA_11C}\\
	R_{1,\mathrm{sec}} <& A(P)-C(P),\Label{FM_OSA_12C}\\
	R_1 \geq & R_{1,\mathrm{sec}}, \Label{FM_OSA_8C}\\
	R_2 + R_1-R_{1,\mathrm{sec}} <& B(P), \Label{FM_OSA_20C}\\
	R_1-R_{1,\mathrm{sec}}  <& B(P)-D(P), \Label{FM_OSA_21C}\\
	R_1  <& A(P)+B(P)-E(P).\Label{FM_OSA_22C}
\end{align}
Consider \eqref{FM_OSA_12C}, \eqref{FM_OSA_8C}, \eqref{FM_OSA_20C} and \eqref{FM_OSA_21C} to remove $R_{1,\mathrm{sec}}.$ We obtain the following constraints:
\allowdisplaybreaks
\begin{align}
	R_2<& B(P), \Label{FM_OSA_23}\\
	D(P) <& B(P), \Label{FM_OSA_24}\\
	C(P)< & A(P), \Label{FM_OSA_25}\\
	R_1 + R_2 <& A(P)+B(P)-C(P), \Label{FM_OSA_26}\\
	R_1  <& A(P)+B(P)-C(P)-D(P), \Label{FM_OSA_27}
\end{align}

Therefore, under the exact strict feasibility conditions $D(P) < B(P)$, $C(P) < A(P)$, and $E(P)<A(P)+B(P),$ the following region is achievable 
\allowdisplaybreaks
\begin{align*}
	R_{1}  <& A(P),\\
	R_{1}  <& A(P)+B(P)-\max\{E(P), C(P)+D(P)\},\\
	R_{2} <& B(P),\\
	R_1 + R_2 <& A(P)+B(P)-C(P).
\end{align*}


\begin{thebibliography}{00}
	
	\bibitem{Shannon1961}
	C. E. Shannon, 
	``Two-way communication channels,'' 
	\emph{Proc. 4th Berkeley Symp. Math. Stat. and Prob.}, vol. 1, pp. 611 -- 644, 1961.
	
	\bibitem{Dueck1979}
	G. Dueck, 
	``The capacity region of the two-way channel can exceed the inner bound,''
	\emph{Inform. Contr.}, vol. 40, no. 3, pp. 258 -- 266, 1979.
	
	\bibitem{Schalkwijk1983}
	J. P. M. Schalkwijk, 
	``On an extension of an achievable rate region for the binary multiplying channel,'' 
	\emph{IEEE Trans. on Inform. Theory}, 
	vol. 29, no. 3, pp. 445 -- 448, 1983.
	
	\bibitem{Han1984} 
	T. S. Han, 
	``A general coding scheme for the two-way channel,'' 
	\emph{IEEE Trans. on Inform. Theory}, vol. 30, no. 1, pp. 35 -- 44,
	1984.
	
	\bibitem{Varshney2013}
	L. R. Varshney, 
	``Two way communication over exponential family type channels,''
	{\em Proc. 2013 IEEE International Symposium on Information Theory}, 
	Istanbul, Turkey, 2013, pp. 2795 -- 2799.
	
	\bibitem{SAL2016}
	L. Song, F. Alajaji, and T. Linder, 
	``Adaptation is useless for two discrete additive-noise two-way channels,'' 
	{\em Proc. 2016 IEEE International Symposium on Information Theory (ISIT)}, 
	Barcelona, Spain, 2016, pp. 1854 -- 1858.
	
	\bibitem{CVA2017}
	A. Chaaban, L. R. Varshney, and M. -S. Alouini, 
	``The capacity of injective semi-deterministic two-way channels,'' 
	{\em Proc. 2017 IEEE International Symposium on Information Theory (ISIT)}, 
	Aachen, Germany, 2017, pp. 431 -- 435.
	
	\bibitem{ZBS1986}
	Z. Zhang, T. Berger, and J. P.M. Schalkwijk, 
	``New outer bounds to capacity regions of two-way channels,'' 
	\emph{IEEE Trans. on Inform. Theory}, 
	vol. 32, no. 3, pp. 383 -- 386, 1986.

	\bibitem{PW1989}
	 A. P. Hekstra and F. M. J. Willems,
	``Dependence balance bounds for	single output two-way channels,'' 
	\emph{IEEE Trans. on Inform. Theory}, 
	vol. 35, no. 1, pp. 44 -- 53, 1989.
	
	\bibitem{TU2007}
	R. Tandon and S. Ulukus, 
	``On Dependence Balance Bounds for Two Way Channels,'' 
	\emph{41st Asilomar Conference on Signals, Systems and Computers}, 
	Pacific Grove, CA, November 2007.
	
	\bibitem{Shannon1949}
	C. E. Shannon, 
	``Communication theory of secrecy systems,'' 
	\emph{Bell Sys. Tech. J.}, vol. 28, pp. 656 -- 715, 1949.
	
	\bibitem{Wyner1975}
	A. Wyner, 
	``The wire-tap channel,'' 
	\emph{Bell Sys. Tech. J.}, vol. 54, pp. 1355 -- 1387, 1975.
	
	\bibitem{src:Csiszar1996}
	I.~Csisz{\'a}r, 
	``Almost independence and secrecy capacity,'' 
	\emph{Probl. Peredachi Inf.}, vol. 32, no. 1, pp. 48 -- 57, 1996.
	
	\bibitem{MH2006}
	M. Hayashi, 
	``General nonasymptotic and asymptotic formulas in channel resolvability and identification capacity and their application to the wiretap channel,''  
	\emph{IEEE Trans. on Inform. Theory}, 
	vol. 52, no. 4, pp. 1562 -- 1575, 2006.
	
	\bibitem{TY2007} 
	E. Tekin and A. Yener, 
	``Achievable rates for two-way wire-tap channels,'' 
	{\em Proc. 2007 IEEE International Symposium on Information Theory}, 
	Nice, France, 2007, pp. 941 -- 945.
	
	\bibitem{TY2008} 
	E. Tekin and A. Yener, 
	``The general Gaussian multiple-access and two-way wire-tap channels: Achievable rates and cooperative jamming,''
	\emph{IEEE Trans. on Inform. Theory}, vol. 54, no. 3, pp. 2735 -- 2751, 2008.
	
	\bibitem{GKYG2013} 
	A. El Gamal, O. O. Koyluoglu, M. Youssef, and H. El Gamal, 
	``Achievable secrecy rate regions for the two-way wiretap channel,'' 
	\emph{IEEE Trans. on Inform. Theory}, 
	vol. 59, no. 12, pp. 8099 -- 8114, 2013.
	
	
	\bibitem{QCHT2016}
	C. Qi, Y. Chen, A. J. H. Vinck, and X. Tang, 
	``One-sided secrecy over the two-way wiretap channel,'' 
	{\em Proc. 2016 International Symposium on Information Theory and Its Applications (ISITA)}, 
	Monterey, CA, USA, 2016, pp. 626 -- 630.
	
	\bibitem{QDT2017}
	C. Qi, B. Dai and X. Tang, 
	``Achieving both positive secrecy rates of the users in two-way wiretap channel by individual secrecy,'' 
	{\emph{CoRR}}, abs/1707.05930, 2017.
	
	\bibitem{PB2011}
	A. J. Pierrot and M. R. Bloch, 
	``Strongly Secure Communications Over the Two-Way Wiretap Channel,'' 
	{\emph{IEEE Transactions on Information Forensics and Security}}, 
	vol. 6, no. 3, pp. 595 -- 605, 2011.
	
	\bibitem{CH2025}
	Y. Chen and M. Hayashi, 
	``Adaptive Coding for Two-Way Wiretap Channel Under Strong Secrecy,'' 
	\emph{Information Theory and Related Fields}, 
	vol. 14620, pp. 243 -- 273, 2025.
	
	\bibitem{HC2023}
	M. Hayashi and Y. Chen, 
	``Non-Adaptive Coding for Two-Way Wiretap Channel with or without Cost Constraints,''  
	\emph{IEEE Trans. on Inform. Theory}, 
	vol. 70, no. 7, pp. 4611 -- 4633, 2024. 
	
	\bibitem{HT16}
	M. Hayashi and M. Tomamichel, 
	``Correlation Detection and an Operational Interpretation of the Renyi Mutual Information,'' 
	{\em Journal of Mathematical Physics}, vol. 57, no. 10, pp. 102201, 2016.
	
	\bibitem{TH}
	M. Tomamichel and {M. Hayashi}, 
	``Operational interpretation of R\'enyi information measures via composite hypothesis testing against product and Markov distributions,'' 
	\emph{IEEE Trans. on Inform. Theory}, 
	vol. 64, no. 2, pp. 1064 -- 1082, 2018. 
	
	
\end{thebibliography}
\end{document}